\RequirePackage{fix-cm}
\documentclass[smallextended]{svjour3}       
\smartqed  
\usepackage{
amsmath,amsfonts, amssymb,epsfig,color,graphicx,enumitem,multicol} 
\begin{document}

\title{Markov chain aggregation and its applications to combinatorial reaction networks
}


\author{
Arnab Ganguly$^*$ \and
Tatjana Petrov$^*$ \and 
Heinz Koeppl
}

\institute{A. Ganguly\at
              Division of Applied Mathematics, Brown University,\\
               182 George St, Providence, USA \\
              Tel.: +1 401 863 7422\\
              Fax: +1 401 863 1355\\
              \email{arnab\textunderscore ganguly@brown.edu}           
\and
              T. Petrov\at
              Automatic Control Lab, Physikstrasse 3, Zurich 8092, Switzerland \\
              Tel.: +41 44 63 29785\\
              Fax: +41 44 632 1211\\
              \email{tpetrov@control.ee.ethz.ch}           
           \and
           H. Koeppl\at
           Automatic Control Lab, Physikstrasse 3, Zurich 8092, Switzerland\\
             Tel.: +41 44 632 7288\\
              Fax: +41 44 632 1211\\
              \email{koepplh@ethz.ch}    
\and
$*$ authors with equal contribution
}

\date{Received: date / Accepted: date}





\newcommand{\chk}[1]{ {\textit{\textcolor{blue}{HK: #1}}} } 
\newcommand{\cag}[1]{ {\textit{\textcolor{magenta}{AG: #1}}} } 
\newcommand{\cm}{\color{magenta}}
\newcommand{\ctp}[1]{ {\textit{\textcolor{red}{TP: #1}}} } 

\newcommand{\EE}{\ensuremath{\mathsf{E}}}
\newcommand{\PP}{\ensuremath{\mathsf{P}}}
\newcommand{\np}{\noindent}
\newcommand{\hs}{\hspace}
\newcommand{\vs}{\vspace}
\newcommand{\g}{\gamma}
\newcommand{\G}{\Gamma}
\newcommand{\e}{\epsilon}
\newcommand{\oo}{\infty}
\newcommand{\Rt}{\longrightarrow}
\newcommand{\rt}{\rightarrow}
\newcommand{\RT}{\Rightarrow}
\newcommand{\Lt}{\longleftarrow}
\newcommand{\lt}{\leftarrow}
\newcommand{\LT}{\Leftarrow}
\newcommand{\LRT}{\Leftrightarrow}
\newcommand{\non}{\nonumber}

\newcommand\sabs[1]{[#1]_{\sim}}

\def\powerset{\cal P}

\def\funth{\phi_{3}}
\def\relth{\sim_{3}}
\def\nbonds{m}
\def\usg{\mathcal G}
\def\edgeset{\SC{E}}
\def\ruleset{\SC{R}}
\def\rule{R}
\def\allsg{\mathbb G}
\def\nAB{n_{AB}}
\def\nBC{n_{BC}}
\def\nABC{n_{ABC}}
\def\funo{\phi_1}
\def\nA{n_A}
\def\nB{n_B}
\def\partspg{\tilde{\allsg}}
\def\eqclass{\mathbb{\s}}
\def\eqcb{\mathbb{B}}
\def\vars{Var}
\def\edgeSet{{\cal E}}
\def\model{M}
\def\fullInterface{\hat{\Sigma}}
\def\sites{{\cal S}}
\def\rel{\sim}
\def\relo{\sim_1}
\def\relt{\sim_2}
\def\st{s}
\def\stX{s}
\def\partsto{[\st]_{\relo}}
\def\partst{[\st]_{\rel}}
\def\partstt{[\st]_{\relt}}
\def\edge{e}
\def\partstti{[\st_i]_{\relt}}
\def\partstoi{[\st_i]_{\relo}}
\def\partsti{[\st_i]_{\rel}}
\def\lab{L}
\def\noReac{m}
\def\noSpec{n}
\def\stoV{\nu} 
\def\stoVB{\nu} 
\def\nodeSet{{\cal V}}
\def\siteSet{{\cal S}}
\def\nodes{{V}}
\def\node{{v}}
\def\edges{{E}}
\def\eval{{\psi}}
\def\embeds{{\models}}
\def\funt{\phi_2}
\def\mABq{m_{AB*}}
\def\mBCq{m_{*BC}}
\def\rule{{R}}

\def\val{{\epsilon}}
\def\site{{s}}
\def\agent{{A}}

\def\allsg{{\mathbb G}}
\def\mAB{m_{AB}}
\def\mBC{m_{BC}}
\def\mABC{m_{ABC}}
\def\funo{\phi_1}

\def\agproct{{Z_t}}
\def\symb{{\zeta}}

\def\rate{{c}}

\def\lhs{{\sitegraph}}
\def\rhs{{\sitegraph'}}

\def\sitegraph{{G}}
\def\site{s}
\def\inst{\eta}
\def\nodeo{A}
\def\nodet{B}
\def\nodeth{C}
\def\nins{n} 

\def\noVarsN{v}
\def\noVars{\hat{s}}
\def\noRules{{r}}

\def\var{\sigma}

\def\sitegraphdef{\sitegraph=(\nodes,\Sigma,\edges)}

\newcommand{\siteA}[3]{{#1}_{#2}^{(#3)}}
\newcommand{\ama}[1]{^{#1}}
\newcommand\embed[2]{#1_{#2}}

\def\am{{^n}}
\def\Nsp{n}
\def\Nreact{v}
\def\change{\nu}
\def\nodeo{u}
\def\nodet{v}
\def\siteo{s}
\def\sitet{s'}
\def\subgraph{\subseteq}

\def\SC{\mathcal}
\def\f{\frac}
\def\l{\lambda}
\def\L{\Lambda}
\def\ot{\otimes}
\def\<{\langle}
\def\>{\rangle}
\def\~{\tilde}
\def\y{\mathbf y^*}
\def\N{\mathbb N}
\def\Z{\mathbb Z}
\def\R{\mathbb R}
\def\H{\mathbb H}
\def\LL{\mathbb L}
\def\K{\mathbb K}
\def\X{\mathbb X}
\def\Y{\mathbb Y}
\def\proc{X}
\def\agproc{Y}
\def\statesp{S}
\def\partsp{\tilde{S}}
\def\transmat{P}
\def\aggtrans{\tilde{P}}
\def\initialdist{\pi}
\def\statdist{\mu}
\def\genmat{Q}
\def\aggen{\tilde{Q}}
\def\unifchain{Z}
\def\uniftrans{M}
\def\pois{\xi}

\def\al{{\alpha}}
\def\s{A}
\def\ra{\rightarrow}
\def\rA#1{\stackrel{{#1}}{\ra}}

\def\longrightharpoonup{\relbar\joinrel\rightharpoonup}
\def\longleftharpoondown{\leftharpoondown\joinrel\relbar}

\def\longrightleftharpoons{
  \mathop{
    \vcenter{
      \hbox{
    \ooalign{
      \raise1pt\hbox{$\longrightharpoonup\joinrel$}\crcr
      \lower1pt\hbox{$\longleftharpoondown\joinrel$}
    }
      }
    }
  }
}

\newcommand{\rates}[2]{\displaystyle
                  {\longrightleftharpoons^{#1}_{#2}}}


\maketitle

\begin{abstract}
We consider a continuous-time Markov chain (CTMC) whose state space is partitioned into aggregates, and each aggregate is assigned a probability measure. 
A sufficient condition for defining a CTMC over the aggregates is presented as a variant of weak lumpability, which also 
characterizes that the measure over the original process can be recovered from that of the aggregated one.
We show how the applicability of de-aggregation depends on the initial distribution.
The application section is a major aspect of the article, where we illustrate 
that the stochastic rule-based models for biochemical reaction networks form an important area for usage of the tools developed in the paper. For the rule-based models, the construction of the aggregates and computation of the distribution over the aggregates are algorithmic.
The techniques are exemplified in three case studies.

\end{abstract}


\section*{Introduction}

The theory of Markov processes has  a wide variety of applications ranging from engineering to biological sciences. In systems biology appropriate Markov processes are used in stochastic modeling of different biochemical reaction systems, especially where the constituent species are present in low abundance. 
Aggregation or lumping of a Markov chain is instrumental in reducing the size of the state space of the chain and in modeling of a partially observable system. 
Typically, the original state space, $\statesp$, of the Markov chain $\{\proc_n\}$ is partitioned into a set of equivalence classes, $\partsp=\{\s_1,\ldots,\s_m\}$, and a process, $\{Y_n\}$, is defined over  $\partsp$.  
More precisely, let $\pi$ be an initial distribution on $\statesp$ for the chain $\{X_n\}$. 
For a given partition $\partsp$ of $\statesp$, let the aggregated chain $\{Y_n\}$ be defined by
$$\{Y_n = \s_m\} \mbox{ if and only if } \{X_n\in \s_m\}.$$
Observe that $\{Y_n\}$ is not necessarily Markov, nor homogeneous. 
Conditions are imposed on the transition matrix of the Markov chain $\{\proc_n\}$ to ensure that the new process $\{Y_n\}$ is also Markov (see \cite{weakLumpDTMC}, 
\cite{weakLumpCTMC}, 
\cite{Sokolova03onrelational},
\cite{buchholz_lump},
\cite{lumpCommutCTMCExact} and references therein). 
In this context, {\em strong lumpability} refers to the property of $\{X_n\}$, when the aggregated process $\{Y_n\}$ (associated with a given partition) is Markov with respect to any initial distribution $\pi$. 
If $P$ denotes the transition matrix of $\{X_n\}$, then it has been shown that a necessary and sufficient condition for $\{X_n\}$ to be {\em strongly lumpable} with respect to the partition $\partsp$ is that
for every $\s_k,\s_l$, $\sum_{s\in \s_l}P(s',s) = \sum_{s\in \s_l} P(s'',s)$ for any $s',s'' \in \s_k$.  Tian and Kannan \cite{lumpCommutCTMCExact} extended the notion of strong lumpability to continuous time Markov chains. 
A more general situation is when $\{X_n\}$ is {\em weakly lumpable} (with respect to a given partition), that is,  when $\{Y_n\}$ is Markov for a subset of initial distributions $\pi$.  
The notion first appeared in \cite{KS60:lump} and subsequent
 papers \cite{weakLumpDTMC,weakLumpCTMC,JLedoux95} focussed toward developing an algorithm for characterizing the desired set of initial distributions. The characterization is done through some kind of recursive equations which sometimes might be hard to read.

The sufficient condition that we provide in the current paper  for $\{X_n\}$ to be weakly lumpable with respect to partition $\tilde{S}$ is easy to read and is geared toward applications in combinatorial reaction networks. In particular, our condition enables us to recover information about the original Markov chain from the smaller aggregated one (see Theorems \ref{thm_sdinv}, \ref{thm_conv}, \ref{thmct_sdinv}, \ref{thmct_conv}). This `invertibility' property is particularly useful for modeling protein networks and is not addressed explicitly for weakly lumpable chains in previous literature.   A variant of our condition can be found in \cite{buchholz_bisimulation} where the author considered  backward bisimulation over a class of weighted automata (finite automata where weights and labels are assigned to transitions).  For each $i$, let $\al_i$ be a probability measure over $A_i$. 
The condition that we impose requires that for every $i$ and $j$,  $\sum_{s\in \s_i} \al_i(s)P(s,s')/\al_j(s'), s' \in \s_j$ is constant over $\s_j$.
The condition can be interpreted as follows: Suppose that you are at the state $s' \in \s_j$ and you look back and try to compute the probability that your immediate previous position was somewhere in $\s_i$.  
The above condition implies that this probability is same no matter where you look back from in $\s_j$. This in particular generalizes the notion of exact lumpability which corresponds to the case when the measures $\alpha_i$ are uniform \cite{buchholz_lump}.
Interestingly, if the initial distribution `$\pi$ respects $\al_i$' in the sense that $\pi(s)/\pi(\s_i) = \al_i(s)$, then the conditional probability $P(X_t=s\mid Y_t = A_i) = \al_i(s), \mbox{ for all } t>0$. In fact, we proved that even if the initial distribution does not respect the $\al_i$, the above result holds asymptotically. These convergence results established in the article are particularly useful for modeling purposes and to the best of our knowledge have not been discussed before. They imply that the modeler can run the `smaller', aggregated process $\{Y_t\}$ and can still extract information about the `bigger' process $\{X_t\}$ if the need arises. This is further illustrated in the application section.

The main practical difficulty in aggregating and de-aggregating a Markov chain is to construct the appropriate partition, and
to find the probability measure over the aggregates. 
Both issues are successfully resolved in the application to the rule-based-models of biochemical reaction networks.

Traditional modeling of biochemical networks is centered around chemical reactions among molecular species and a state of a network is a multi-set of molecular species. 
A species can be, for instance, a protein or its phosphorylated form or a protein complex that consists of several proteins bound to each other.
Especially, in cellular signal transduction the number of different such species can be combinatorially large, due to the rich internal structure of proteins and their mutual binding \cite{hlavacekws_2003},\cite{walshct_2006}. 
For example, one model of the early signaling events in epidermal growth factor receptor (EGFR) network, with only $8$ different proteins gives rise to $2748$ different molecular species \cite{bli:06}. 
In such cases, a formal description of the cellular process using different reactions and species becomes computationally expensive. 

Instead, an efficient way to encode different molecular interactions is to use  a {\em site-graph} based model.
 A \emph{site-graph} is a generalization of a graph where each node contains different types of sites, and edges can emerge from these sites. Molecular species are often suitably represented by site-graphs, where nodes are proteins and their sites are the protein binding-domains or modifiable residues; the edges indicate bonds between proteins.   
Every species is a connected site-graph, and in accordance with the traditional model, a state of a network is a multi-set of connected site-graphs. 
Importantly, more detailed description of the species' structure allows to describe interactions locally, between parts of molecular species (sometimes refered to as  \emph{fragments}). 
For instance, it can be stated by one {\em rewrite rule}, that any species containing a protein of type $A$ can have that protein $A$ phosphorylated. 
In this case, the event of phosphorylation of $A$ is independent of the rest of the species' context, i.e. $A$ can equally be part of a dimer (complex of two proteins) or of a very large protein complex. 
It is precisely this independence between the molecular events  we exploit when aggregating states and constructing a suitable aggregated process. 
In the present article we present a rigorous construction of a Markov chain $\{X_t\}$ on an appropriate space of site-graphs which essentially tells us how the `reaction soup' looks like at different points of time. It is then shown that the usual species-based Markov chain can be constructed as an aggregation of $\{X_t\}$. But more importantly, there exist other aggregations which lead to Markov chains living on much smaller state spaces, and information about the species-based model can be extracted at any point of time from these smaller Markov chains (see Theorem \ref{thm:conn}).

One important feature of the presented application is that it provides an effective way of constructing the partition and the accompanying distributions over the aggregates  (also see \cite{lumpability}, \cite{sasb2011}).
In particular, the three case studies presented at the end exemplify our approach to effectively reduce the state-space of the CTMCs in the context of molecular interactions (see Table \ref{tab:cs3} for an overview of the achieved reduction).  

The work presented in this paper is inspired by the related work of 
\cite{lumpability},
where 
the general algorithm for reducing the stochastic behavior of any Kappa \cite{Kappa} model is shown. 
The proof uses a cumbersome object of a weighted labeled transition system, supplied with all the details which are necessary when providing a general reduction algorithm. 
In contrast, in the present article the mathematical treatment of the rule-based models has been carried out efficiently by using the tools of graph theory. The analysis of Markov chain aggregation is done for general Markov chains  whose application covers but is certainly not limited to the class of rule-based models.
In such a set-up, the existing reduction framework is extended with a criterion on the rule-set for claiming the asymptotic possibility of reconstruction of the species-based dynamics. 

The rest of the work is organized as follows. 
In Section \ref{sec:weak}, we describe conditions on the transition matrix and initial distribution of the Markov chain $\{X_n\}$ which will ensure that the aggregated chain $\{\agproc_n\}$ is also Markov. The conditions described are tailor-made for our applications to biochemical reaction networks. We also prove convergence properties of the transition probabilities of the aggregated chain when the initial distribution does not satisfy the required conditions. The case of continuous time chains has been treated in Section \ref{sec:cont}. Section \ref{sec:app} first discusses the traditional Markov chain modeling of biochemical reaction systems using reactions and species. Next, the mathematical definition of site-graphs is introduced and the formal description of site-graph based modeling of protein-protein interaction is given. 
Section \ref{sec:application} is devoted to applications. 
We describe the criteria for testing the aggregation conditions on the CTMCs which underly rule-based models.
Illustrative case studies are given at the end. 

\section{Discrete time case}
\label{sec:weak}
Let $\{\proc_n\}$ be a Markov chain taking values in a 
finite set $\statesp$ with transition matrix $\transmat$ and initial probability distribution $\initialdist$. 
Let $\partsp=\{\s_1,\ldots,\s_m\}$ be a finite partition of $\statesp$.
Moreover, let $\{\al_i\}_{i=1,\ldots,m}$ be a family of probability measures on $\statesp$, such that $\al_i(s)=0$ for $s\notin \s_i$.
Define $\delta:\partsp\times \statesp\ra \R_{\geq 0}$ by 
\begin{align*}
\label{dfn:Delta}
\delta(\s_i,s)= 
\frac{\sum_{s'\in \s_i} \al_i(s')\transmat(s',s)}{\al_j(s)}, \hbox{ where } s\in \s_j. 
\end{align*}
Assume that the following condition holds.
\begin{enumerate}[label={(Cond\arabic*)}, leftmargin=*, align=right]
\item\label{c1} For any $\s_i,\s_j\in \partsp$ and $s,s' \in \s_j$,
$\delta(\s_i,s)=\delta(\s_i,s')$.
\end{enumerate}
Fix $s \in \s_j$ and let $\aggtrans(\s_i,\s_j) := \delta(\s_i,s)$. Notice that $\aggtrans$ is  unambiguously defined under \ref{c1}.


\begin{theorem}\label{thm_prob} \rm
$\aggtrans$ is a probability transition matrix.
\end{theorem}

\begin{proof}
Notice that by \ref{c1},
$$\s_j(s)\aggtrans(\s_i,\s_j) =\sum_{s'\in \s_i} \al_i(s')\transmat(s',s). $$
Summing over $s \in \s_j$, we have
$$\aggtrans(\s_i,\s_j)=\sum_{s'\in \s_i}\sum_{s\in \s_j} \al_i(s')\transmat(s',s).$$
It follows that
\begin{align*}
\sum_{j=1}^m \aggtrans(\s_i,\s_j) &=\sum_{j=1}^m \sum_{s'\in \s_i}\sum_{s\in \s_j} \al_i(s')\transmat(s',s)\\
& = \sum_{s'\in \s_i}\al_i(s') \sum_{s \in \statesp} \transmat(s',s) = \sum_{s'\in \s_i}\al_i(s') \\
& =1.
\end{align*}

\end{proof}

\begin{definition} \rm
For any probability distribution $\pi$ on $\statesp$,  define the probability distributions $\pi|_{\s_i}$ on $\s_i$ and $ \tilde{\pi}$ on $\partsp$ by 
\begin{align*}
\pi|_{\s_i}(s) := \frac{\pi(s)}{\sum_{s'\in\s_i}\pi(s')}, \ \
\tilde{\pi}(\s_i):=\sum_{s'\in\s_i}\pi(s').
\end{align*}
\end{definition}

\begin{definition}
We say that a probability distribution  $\pi$ {\bf respects} $\{\al_i:i=1,\ldots, m\}$ if  $\initialdist|_{\s_i}(s) = \al_i(s)$ for $s\in \s_i, i=1,\ldots,m$.
\end{definition}

\subsection{Aggregation and de-aggregation}

Throughout, we will assume that  $\{\agproc_n\}$ is a Markov chain taking values in $\partsp$ with transition matrix $\aggtrans$ and initial distribution $\tilde{\initialdist}$.

\begin{theorem} \rm \label{thm_sdinv}
 Assume that $\initialdist$ respects $\{\al_i:i=1,\ldots, m\}$. Then for all $n=0,1,\ldots$
\begin{enumerate}[label={\rm (\roman*)}, leftmargin=*, align=right]
\item \label{snd} (lumpability)  
$\displaystyle{
\PP(Y_n=\s_i) =  \PP(X_n\in \s_i); 
}$
\item \label{inv} (invertibility)   $\displaystyle{\PP(X_n=s) = \PP(Y_n=\s_i) \al_i(s)}$.
\end{enumerate}
\end{theorem}

We need the following two lemmas to prove Theorem \ref{thm_sdinv}.
\begin{lemma} \rm \label{lem1}
Assume that for all $i=1,\ldots,m$, $ \PP(X_{n-1}=s | X_{n-1}\in \s_i) =\initialdist P^{n-1}|_{\s_i}(s)= \al_i(s) $. 
Then 
$\PP(X_n\in \s_j | X_{n-1} \in \s_i) = \aggtrans(\s_i,\s_j).$
\end{lemma}

\begin{proof} Notice that
\begin{align*} 
\PP(X_n\in \s_j | X_{n-1} \in \s_i) & = \frac{\PP(X_n\in \s_j, X_{n-1} \in \s_i)}{\PP(X_{n-1}\in A_i)}  \\
& = \frac{\sum_{s'\in\s_i}\sum_{s\in \s_j}\PP(X_n= s, X_{n-1} = s')}{\PP(X_{n-1}\in A_i)}  \\
& = \frac{\sum_{s'\in\s_i}\sum_{s\in \s_j}\PP(X_{n-1}=s')\PP(X_n= s | X_{n-1} = s')}{\PP(X_{n-1}\in A_i)} \\
& = \frac{\sum_{s'\in\s_i}\sum_{s\in \s_j}\PP(X_{n-1}\in \s_i)\PP(X_{n-1}=s' | X_{n-1}\in \s_i) \transmat(s',s)}{\PP(X_{n-1}\in A_i)} \\
& = \frac{\PP(X_{n-1}\in \s_i) \sum_{s'\in\s_i}\sum_{s\in \s_j} \al_i(s') \transmat(s',s)}{\PP(X_{n-1}\in A_i)}, \ \mbox{ by the hypothesis} \\
& = \sum_{s'\in\s_i}\sum_{s\in \s_j}\al_i(s') \transmat(s',s) \frac{\al_j(s)}{\al_j(s)}\\
& =\sum_{s\in \s_j} \aggtrans(\s_i,\s_j) \al_j(s), \ \mbox{ by the definition of }\aggtrans \\
& = \aggtrans(\s_i,\s_j) \sum_{s\in \s_j} \al_j(s) =  \aggtrans(\s_i,\s_j).
\end{align*}
%
\end{proof}

\begin{lemma} \rm \label{lem2}  Assume that $\initialdist|_{\s_i}(s) = \al_i(s)$ for $s\in \s_i, i=1,\ldots,m$. Then
$\displaystyle{\initialdist P^n|_{\s_i}(s) = \PP(X_n=s | X_n\in \s_i) = \al_i(s)}.$
\end{lemma}

\begin{proof}
The first equality is of course by the definition. For the second, we use induction. 
The case $n=0$ is given. 
Suppose that the statement holds for $k=n-1$. First observe that if $s\notin \s_i$, then both sides equal $0$. So assume that $s\in \s_i$.
Then, by Lemma \ref{lem1}, we have that $\PP(X_n\in \s_j | X_{n-1} \in \s_i) = \aggtrans(\s_i,\s_j)$.
Next note that
\begin{align*}
P(X_n = s | X_n \in \s_i) 
& = \frac{\PP(X_n=s)}{\PP(X_n\in \s_i)} \\
& = \frac{\sum_{s'\in S}\PP(X_{n-1}=s', X_n=s)}{\PP(X_n\in \s_i)} \\
& = \frac{\sum_{s'\in S}\PP(X_{n-1} = s')\PP(X_n=s|X_{n-1}=s')}{\PP(X_n\in \s_i)} \\
& = \frac{\sum_{j=1}^m\sum_{s'\in \s_j}\PP(X_{n-1} = s') \PP(X_n=s|X_{n-1}=s')}{\PP(X_n\in \s_i)} \\
& = \frac{\sum_{j=1}^m\sum_{s'\in \s_j}\PP(X_{n-1} \in \s_j) \PP(X_{n-1}=s'|X_{n-1}\in \s_j)\transmat(s',s)}{\PP(X_n\in \s_i)} \\
& = \frac{\sum_{j=1}^m\PP(X_{n-1}\in \s_j)(\sum_{s'\in \s_i} \al_j(s')\transmat(s',s)\cdot \frac{\al_i(s)}{\al_i(s)})}{\PP(X_n\in \s_i)}\\
& = \frac{\sum_{j=1}^m \PP(X_{n-1}\in \s_j)\aggtrans(\s_j,\s_i) \al_i(s)}{\PP(X_n\in \s_i)}\\
& = \al_i(s)\frac{\sum_{j=1}^m \PP(X_{n-1}\in \s_j) \PP(X_n\in \s_i | X_{n-1}\in \s_j)}{\PP(X_n\in \s_i)}\\
& = \al_i(s)\frac{\PP(X_n\in \s_i)}{\PP(X_n\in \s_i)}  =\al_i(s).
\end{align*}
\end{proof}

We next proceed to prove Theorem \ref{thm_sdinv}.

\begin{proof} (Theorem \ref{thm_sdinv})
We use induction. Notice that  both the statements hold for $n=0$. Assume that \ref{snd} and \ref{inv} hold for $n-1$.
Then $ \PP(X_{n-1}=s | X_{n-1}\in \s_i)  =\al_i(s)$, and hence by Lemma \ref{lem1} $\PP(X_n\in \s_j | X_{n-1} \in \s_i) = \aggtrans(\s_i,\s_j).$
Therefore,
\begin{align*}
\PP(\agproc_n=\s_i) 
& = \sum_{j=1}^m \PP(\agproc_{n-1}=\s_j) \aggtrans(\s_j,\s_i) \\
& = \sum_{j=1}^m \PP(\proc_{n-1}\in\s_j)\PP(\proc_n\in \s_j | \proc_{n-1} \in \s_i) \\
& = \PP(\proc_n\in\s_i).
\end{align*}
This proves \ref{snd}.
Next, notice that Lemma \ref{lem2} implies 
\begin{align*}
\PP(X_n=s) &= \al_i(s) \PP(\proc_n \in \s_i) \\
& = \al_i(s)\PP(Y_n=\s_i), \mbox{ by } \ref{snd}.
\end{align*}
This proves \ref{inv}.

\end{proof}

\begin{remark}
Notice that we have proved that under the assumption  $\initialdist|_{\s_i}(s) = \al_i(s)$,  $\PP(\proc_n\in \s_j | \proc_{n-1} \in \s_i) = \aggtrans(\s_i,\s_j),$ for $n=1,2,\ldots$.
\end{remark}

\subsection{Convergence}
In the previous section, we proved that if $\{X_n\}$ is a discrete time Markov chain on $\statesp$ with  initial distribution $\initialdist$ respecting $\{\al_i:i=1,\ldots, m\}$, then the aggregate process $\{Y_n\}$ is an aggregated Markov chain satisfying lumpability and invertibility property. We now investigate the case when the initial distribution of $\{X_n\}$ doesn't respect $\{\al_i:i=1,\ldots, m\}$. We start with the following theorem.
\begin{theorem} \rm
\label{thm_transmat}
$\displaystyle{\aggtrans^n(\s_i,\s_j) =\frac{\sum_{s'\in \s_i} \al_i(s')\transmat^n(s',s)}{\al_j(s)} }$, for any $s\in \s_j$.
\end{theorem}

\begin{proof}
We use induction. Notice that for $n=1$, the assertion is true by the definition of $\aggtrans$. Assume that the statement holds for some $n$. Then,
\begin{align*}
\aggtrans^{n+1}(\s_i,\s_j)& =\sum_k \aggtrans(\s_i,\s_k)\aggtrans^{n}(\s_k,\s_j)\\
& = \sum_k \left(\sum_{s'\in \s_i} \frac{\al_i(s')\transmat(s',s_0)}{\al_k(s_0)}\right)\left(\sum_{s''\in \s_k} \frac{\al_k(s'')\transmat^n(s'',s)}{\al_j(s)}\right), \ \mbox {for any } s_0 \in \s_k\\
& = \sum_k \sum_{s''\in \s_k}\left(\sum_{s'\in \s_i} \frac{\al_i(s')\transmat(s',s_0)}{\al_k(s_0)}\right)\left(\frac{\al_k(s'')\transmat^n(s'',s)}{\al_j(s)}\right)\\
&=\sum_k \sum_{s''\in \s_k}\left(\sum_{s'\in \s_i} \frac{\al_i(s')\transmat(s',s'')}{\al_k(s'')}\right)\left(\frac{\al_k(s'')\transmat^n(s'',s)}{\al_j(s)}\right), \  \mbox { by \ref{c1}} \\
& =\frac{1}{\al_j(s)}\sum_k \sum_{s''\in \s_k}\left(\sum_{s'\in \s_i} \al_i(s')\transmat(s',s'')\transmat^n(s'',s)\right)\\
& =\frac{1}{\al_j(s)} \sum_{s'\in \s_i}\al_i(s') \sum_k\sum_{s''\in \s_k}\transmat(s',s'')\transmat^n(s'',s)\\
& =\frac{\sum_{s'\in \s_i} \al_i(s')\transmat^{n+1}(s',s)}{\al_j(s)} 
\end{align*}
\end{proof}
We say that $s\rt s'$, if for some $n \geq 0$, $\transmat^n(s,s')>0$. Recall that the Markov chain $\{\proc_n\}$ is irreducible if $s\rt s'$ for any $s,s' \in \statesp$. 
One corollary of Theorem \ref{thm_transmat} is that if for $s\in \s_i, s'\in \s_j$, $s\rt s'$ then $\s_i \rt \s_j$ for the Markov chain $\agproc$. In fact, we have the following result.

\begin{theorem} \rm\label{thm_prop}
Let $\{\proc_n\}$ be a discrete time Markov chain on $\statesp$ with transition probability matrix $\transmat$ and $\{\agproc_n\}$  a Markov chain taking values in $\partsp$ with transition matrix $\aggtrans$. Then
\begin{enumerate}[label={\rm (\roman*)}, leftmargin=*, align=right]
\item \label{c_irred} If the process $\{\proc_n\}$ is irreducible, then so is $\{\agproc_n\}$.
\item \label{c_rec} If $s\in \s_i$ is recurrent for the process $\{\proc_n\}$, then so is $\s_i$ for the process $\{\agproc_n\}$.
\item \label{c_ap}If $s\in \s_i$ has period $1$, then the period of $\s_i$ is also $1$.
\end{enumerate}
\end{theorem}

\begin{proof}
\ref{c_irred} is immediate from the discussion.
Notice that if $s \in \s_i, s'\in \s_j$, then by Theorem \ref{thm_transmat}, $\aggtrans^n(\s_i,\s_j) \geq \frac{\al_i(s)}{\al_j(s')} \transmat^n(s,s').$
Therefore, it follows that if $s \in \s_i=\s_j$, then $\aggtrans^n(\s_i,\s_i) \geq \transmat^n(s,s)$. 
Now $s\in \s_i$ is recurrent if and only if $\sum_n\transmat^n(s,s) = \infty$. 
Observe that 
$$\sum_n \aggtrans^n(\s_i,\s_i) \geq \sum_n\transmat^n(s,s)  =\infty.$$
It follows that for the process $\{\agproc_n\}$, $\s_i$ is recurrent. 
This proves \ref{c_rec}. Next observe that 
$$\{n: \transmat^n(s,s) > 0\} \subset \{n: \aggtrans^n(\s_i,\s_i) > 0\},$$
and \ref{c_ap} follows immediately.

\end{proof}

For the following results we will assume that there exists a probability distribution  $\initialdist$  on $\statesp$ which respects $\{\al_i:i=1,\ldots, m\}$.
\begin{theorem}\rm\label{thm_stat}
Let $\{X_n\}$ be a discrete time Markov chain on $\statesp$ with transition probability matrix $\transmat$ and unique stationary distribution
$\statdist$. Then $\statdist$ respects $\{\al_i:i=1,\ldots, m\}$.
\end{theorem}

\begin{proof}
Let $\initialdist$ be a probability distribution on $\statesp$ which respects $\{\al_i:i=1,\ldots, m\}$. 
Now since $\statdist$ is unique, we have for any set $A$, $\frac{1}{n}\sum_{k=1}^n \initialdist \transmat^k(A) \rt \statdist(A)$ (see \cite{ref:Her-LerLas-03}).
By the choice of $\initialdist$, $\initialdist(s) = \al_i(s)\initialdist(\s_i)$ for $s\in\s_i$. By Theorem \ref{thm_sdinv}, $ \initialdist \transmat^k(s) =\al_i(s)  \initialdist \transmat^k(\s_i), s \in \s_i$. 
Therefore, it follows that $\frac{1}{n}\sum_{k=1}^n \initialdist \transmat^k(s)= \al_i(s)\frac{1}{n}\sum_{k=1}^n\initialdist \transmat^k(\s_i), s\in \s_i.$
Taking limit as $n\rt \infty$,  it implies that 
$\mu|_{\s_i} (s) = \al_i(s), s \in \s_i, i=1,\ldots, m.$

\end{proof}

For any set $A\subset \statesp$, let  $\transmat^{(n)}(s,A) := \frac{1}{n}\sum_{k=1}^n\transmat^k(s,A)$,
 and $\PP^{(n)}(\proc_n \in A) := \frac{1}{n}\sum_{k=1}^n \PP(\proc_k \in A)$.
\begin{theorem} \rm \label{thm_conv}
Let $\{\proc_n\}$ be an irreducible  Markov chain taking values in $\statesp$ with transition matrix $\transmat$. Let $\statdist$ be the stationary distribution of $\transmat$. Let $\{\agproc_n\}$ be a Markov chain on $\partsp$ with transition matrix $\aggtrans$. Then $\tilde{\statdist}$ is the stationary distribution for $\aggtrans$. Also for all $n=0,1,\ldots$,
\begin{enumerate}[label={\rm (\roman*)}, leftmargin=*, align=right]
\item \label{snd_asymp} 
$\displaystyle{
\PP^{(n)}(Y_n=\s_i) -  \PP^{(n)}(X_n\in \s_i) \rt 0; 
}$
\item \label{inv_asymp}$\displaystyle{\PP^{(n)}(X_n=s) / \PP^{(n)}(Y_n=\s_i)\rt \al_i(s)}$.
\end{enumerate}
\end{theorem}

\begin{proof}
We first show that $\tilde{\statdist}$ is a stationary distribution of $\aggtrans$. Towards this end, first observe that by Theorem \ref{thm_stat} $\statdist$ respects $\{\al_i:i=1,\ldots, m\}$. Take $\statdist$ as the initial distribution of $\{X_n\}$. Then by \ref{snd} of Theorem \ref{thm_sdinv}, 
$$\tilde{\statdist}(\s_i) =\statdist(\s_i) = \statdist \transmat ^n(\s_i) = \tilde{\statdist} \aggtrans^n(\s_i).$$ 
It follows that $\tilde{\statdist}$ is stationary for $\aggtrans$. Since $\aggtrans$ is irreducible by Theorem \ref{thm_prop}, $\tilde{\statdist}$ is unique.
Now let $\initialdist$ be any initial distribution for $\transmat$. Since  $\tilde{\statdist}$ is the unique stationary distribution for $\aggtrans$, 
$\initialdist \transmat ^{(n)}(\s_i) \rt \tilde{\statdist}(\s_i)$. Hence \ref{snd_asymp} and \ref{inv_asymp} follow.
\end{proof}

The above result can be improved if we assume in addition that the Markov chain $\{X_n\}$ is aperiodic.
\begin{theorem} \rm \label{thm_conv2}
Let $\{\proc_n\}$ be an irreducible, aperiodic  Markov chain taking values in $\statesp$ with transition matrix $\transmat$. Let $\statdist$ be the stationary distribution of $\transmat$. Let $\{\agproc_n\}$ be a Markov chain on $\partsp$ with transition matrix $\aggtrans$. Then 
\begin{enumerate}[label={\rm (\roman*)}, leftmargin=*, align=right]
\item \label{snd_asymp2} 
$\displaystyle{
\PP(Y_n=\s_i) -  \PP(X_n\in \s_i) \rt 0; 
}$
\item \label{inv_asymp2}$\displaystyle{\PP(X_n=s) / \PP(Y_n=\s_i)\rt \al_i(s)}$.
\end{enumerate}
\end{theorem}

\begin{proof}
By Theorem \ref{thm_prop}, the Markov chain $\{\agproc_n\}$ is also aperiodic and irreducible. Moreover by the previous theorem,
$\tilde{\statdist}$ is the unique stationary distribution for $\{\agproc_n\}$. The result follows by noting that for any aperiodic, irreducible Markov chain $\{Z_n\}$ with a stationary distribution $\eta$, $\PP(Z_n \in A) \rt \eta(A)$.
\end{proof}


\section{Continuous time case}
\label{sec:cont}
We now consider a  continuous time Markov chain, $\{\proc_t\}_{t\in [0,\infty)}$, taking values in a countable set $\statesp$. Let $\genmat$ be the generator matrix for $\{\proc_t\}$. As before, let $\partsp=\{\s_1,\ldots,\s_m\}$ be a finite partition of $\statesp$ and $\{\al_i\}$ be a family of probability measures on $\statesp$ with $\al_i(s) =0, $for $s\notin \s_i$. Define $\Delta:\partsp\times \statesp\ra \R_{\geq 0}$ by 
\begin{equation}\label{defn:cont}
\Delta(\s_i,s)= 
\frac{\sum_{s'\in \s_i} \al_i(s')\genmat(s',s)}{\al_j(s)}, \hbox{ where } s\in \s_j. 
\end{equation}
Assume the following condition holds.
\begin{enumerate}
[label={(Cond\arabic*)}, start=2, leftmargin=*, align=right]
\item\label{c2} For any $\s_i,\s_j\in \partsp$ and $s,s' \in \s_j$,
$\Delta(\s_i,s)=\Delta(\s_i,s')$.
\end{enumerate}
Fix $s \in \s_j$ and let $\aggen(\s_i,\s_j) := \Delta(\s_i,s)$. Notice that $\aggen$ is  unambiguously defined under \ref{c1}.

\begin{theorem}\rm\label{thm_gen}
$\aggen$ is a generator matrix.
\end{theorem}

\begin{proof}
We only need to prove that $\sum_{j=1}^m \aggen(\s_i,\s_j)=0$. The proof proceeds almost exactly in the same way as that of Theorem \ref{thm_prob}.
\end{proof}

For any generator matrix $Q = (q_{ij})$, define
$$q_i  = -q_{ii} = \sum_j q_{ij}.$$

\subsection{Aggregation and de-aggregation}

We next prove the analogue of Theorem \ref{thm_sdinv}.
\begin{theorem}\rm\label{thmct_sdinv}
Let $\{\proc_t\}$ be a continuous time Markov chain taking values in a countable set $\statesp$ with generator matrix $\genmat$ and initial probability distribution $\initialdist$. Let $\{\agproc_t\}$ be a continuous time Markov chain taking values in $\partsp$ with generator matrix $\aggen$ and initial distribution $\tilde{\initialdist}$. Assume that $\initialdist$ respects $\{\al_i:i=1,\ldots, m\}$.  Also assume that there exists an $r>0$ such that $\sup_i q_i <r$. Then for all $t\geq 0$
\begin{enumerate}[label={\rm (\roman*)}, leftmargin=*, align=right]
\item \label{sndct} (lumpability)  
$\displaystyle{
\PP(\agproc_t=\s_i) =  \PP(\proc_t\in \s_i); 
}$
\item \label{invct} (invertibility)   $\displaystyle{\PP(\proc_t=s) = \PP(\agproc_t=\s_i) \al_i(s)}$.
\end{enumerate}
\end{theorem}

We prove the above theorem by constructing a uniformized discrete time Markov chain out of $\{X_t\}.$ For any matrix $A = ((a_{ij}))_{i,j\in \statesp}$, we use the norm $\|A\| = \sup_i \sum_j |a_{ij}|.$ Note by the assumptions in Theorem \ref{thmct_sdinv}, $\|\genmat\| <r <\infty.$
If $\transmat$ denotes the transition probability matrix of $\{\proc_t\}$, then $\transmat$ satisfies the Kolmogorov forward equation
$$\transmat'(t) = \transmat(t)\genmat, t>0.$$
Since $\|\genmat\|<\infty$, the solution to the above equation is given by 
$$\transmat(t) = e^{\genmat t} = \sum_{k=0}^\infty (\genmat t)^k/k!.$$
Define the transition matrix $\uniftrans$ by $M =I + \genmat/r$. Writing $\genmat = r(\uniftrans -I)$ we have
\begin{align}
\label{unifprob}
\transmat(t) = e^{r(\uniftrans-I)t} = e^{-rt}\sum_{k=0}^\infty \frac{(rt)^k}{k!} \uniftrans^k.
\end{align}
Let $\{\unifchain_n\}$ be a Markov chain on $\statesp$ with transition probability matrix $\uniftrans$. Let $\pois$ be a Poisson process with intensity $r$ independent of $\{\unifchain_n\}$. Then \eqref{unifprob} implies that $\{\proc_t\} \stackrel{d} = \{\unifchain(\pois(t))\}$.
We will need to  consider the aggregate Markov chain $\{\tilde{\unifchain}_n\}$ on $\partsp$ with the transition matrix defined by
\begin{align}
\label{agg_unifmat}
\tilde{\uniftrans}(\s_i,\s_j) = \frac{\sum_{s'\in \s_i} \al_i(s')\uniftrans(s',s)}{\al_j(s)},  s\in \s_j.
\end{align}

\begin{lemma}
$\tilde{\uniftrans}$ is well-defined.
\end{lemma}
\begin{proof}
We need to show that $\tilde{\uniftrans}(\s_i,\s_j)$ does not depend on the choice of $s\in \s_j$. Let $\rho(\s_i,s)$ denote the right side of \eqref{agg_unifmat} for $s\in \s_j$. We will use \ref{c2}. First assume that  $i\neq j$. 
Then,
\begin{align}
\label{eq1}
\rho(\s_i,s) = \frac{1}{r}\frac{\sum_{s'\in \s_i} \al_i(s')\genmat(s',s)}{\al_j(s)} = \frac{1}{r}\aggen(\s_i,\s_j),  s\in \s_j.
\end{align}
by the definition of the $\aggen$ matrix. 
If $i=j$, then
\begin{align}
\non
\rho(\s_i,s)& = \frac{\sum_{s'\neq s, s'\in \s_i} \al_i(s')\genmat(s',s)/r+  \al_i(s)(1+\genmat(s,s)/r)}{\al_i(s)}\\
\non
& = \frac{\sum_{s'\in \s_i} \al_i(s')\genmat(s',s)/r - \al_i(s)\genmat(s,s)/r  +  \al_i(s)(1+\genmat(s,s)/r)}{\al_i(s)}\\
& = 1+ \frac{1}{r}\frac{\sum_{s'\in \s_i} \al_i(s')\genmat(s',s)}{\al_i(s)} = 1+\frac{1}{r}\aggen(\s_i,\s_i).
\label{eq2}
\end{align}
It follows $\rho(\s_i,s)$ is independent of the choice of $s\in \s_j, j=1,\ldots,m$.

\end{proof}

We now prove the following commutativity relation. 
\begin{theorem}\rm\label{thm_comm}
Let $\genmat$ be a generator matrix with  $\sup_i q_i <r$ for some $r$, and let
\begin{enumerate}[label={\rm (\roman*)}, leftmargin=*, align=right]
\item $\{\proc_t\}$ be a continuous time Markov chain taking values in a countable set $\statesp$ with generator matrix $\genmat$ and initial probability distribution $\initialdist$. 
\item $\{\agproc_t\}$ be a continuous time Markov chain taking values in $\partsp$ with generator matrix $\aggen$ and initial distribution $\tilde{\initialdist}$.
\item $\{\unifchain_n\}$ be the uniformized  discrete time Markov chain (corresponding to $\{\proc_t\}$) on $\statesp$ with transition matrix $\uniftrans=I+\genmat/r$ and initial distribution $\initialdist$.
\item $\{\tilde{\agproc}_n\}$ be the uniformized discrete time chain (corresponding to $\{\agproc_t\}$) on $\partsp$ with transition matrix
$\bar{\uniftrans} = I +\aggen/r$ and initial distribution $\tilde{\initialdist}$.
\item $\{\tilde{\unifchain}_n\}$ be the discrete time Markov chain on $\partsp$ with transition matrix $\tilde{\uniftrans}$ and initial distribution $\tilde{\initialdist}$.
\end{enumerate}
Then $\{\tilde{\unifchain}_n\} \stackrel{d}=\{\tilde{\agproc}_n\}$.
\end{theorem}
\begin{proof}
We only need to show that $\tilde{\uniftrans} = \bar{\uniftrans}$.  But this readily follows from \eqref{eq1} and \eqref{eq2}.
\end{proof}

\begin{proof}{(Theorem \ref{thmct_sdinv})}
Note that \eqref{unifprob} implies.
\begin{align*}
\PP(\agproc_t = \s_i)& = \sum_{k\geq 0}\PP(\tilde{\agproc}_k = \s_i)\frac{e^{-rt}(rt)^k}{k!}\\
& = \sum_{k\geq 0}\PP(\tilde{\unifchain}_k = \s_i)\frac{e^{-rt}(rt)^k}{k!}\\
&=\sum_{k\geq 0}\PP(\unifchain_k \in \s_i)\frac{e^{-rt}(rt)^k}{k!}\\
& =\sum_{s\in \s_i}(\sum_{k\geq 0}\PP(\unifchain_k =s)\frac{e^{-rt}(rt)^k}{k!})\\
& = \sum_{s\in \s_i} \PP(\proc_t =s) = \PP(\proc_t \in \s_i).
\end{align*}
Here, the second equality is by Theorem \ref{thm_comm} while the third is by \ref{snd} of Theorem \ref{thm_sdinv}. This proves \ref{sndct} and \ref{invct} follows similarly.
\end{proof}

\subsection{Convergence}
Let $\statdist$ be a stationary distribution of the continuous time Markov chain $\{X_t\}$, that is $\statdist$ satisfies $\statdist Q = 0$.
Then we have the corresponding analogue of Theorem \ref{thm_conv}.

\begin{theorem}\rm\label{thmct_conv}
Let $\{\proc_t\}$ be an irreducible  Markov chain taking values in $\statesp$ with generator matrix $\genmat$. Assume that $\sup_i q_i<r$, for some $r>0$. Let $\statdist$ be the stationary distribution of $\genmat$. Let $\{\agproc_t\}$ be a Markov chain on $\partsp$ with generator matrix $\aggen$. Then $\tilde{\statdist}$ is the stationary distribution for $\aggen$. Moreover,
\begin{enumerate}[label={\rm (\roman*)}, leftmargin=*, align=right]
\item \label{snd_asymp} 
$\displaystyle{
\PP(\agproc_t=\s_i) -  \PP(\proc_t\in \s_i) \rt 0; 
}$
\item \label{inv_asymp}$\displaystyle{\PP(\proc_t=s) / \PP(\agproc_t=\s_i)\rt \al_i(s)}$.
\end{enumerate}
\end{theorem}

\begin{proof}
We first consider the uniformized chain $\{\unifchain_n\}$ corresponding to $\{\proc_t\}$ with transition matrix $\uniftrans=I+\genmat/r.$
Note that $\statdist$ is the stationary distribution for $\{\uniftrans\}$. 
It follows by Theorem \ref{thm_conv}, that $\tilde{\statdist}$ is the stationary distribution for $\{\tilde{\unifchain}_n\}$, hence for $\{\tilde{\agproc}_n\}$. It follows that $\tilde{\statdist}\aggen = 0.$ 
Next  $\sup_i q_i <\infty$ guarantees that the chain does not explode. The result follows by noting that for any irreducible, non-exploding continuous time Markov chain $\{Z_t\}$ with a stationary distribution $\eta$, $\PP(Z_t \in A) \rt \eta(A)$ as $t \rt \infty$.

\end{proof}


\section{Formalism} \label{sec:app}

The standard model of biochemical networks is typically based on counting chemical species (complexes). However, for our purpose it is useful to consider a {\em site-graph} based description of the model. We start by briefly outlining the Markov chain formulation of a species-based model of a biochemical reaction system, and then move on to the concept of site-graph.

\subsection{Modeling biochemical networks by a CTMC}

A biochemical reaction system involves multiple chemical reactions and several species. In general, chemical reactions in single cells occur far from thermodynamic equilibrium and the number of molecules of chemical species is often low \cite{ref:Kei87}, \cite{ref:Gup95}. Recent advances in real-time single cell imaging, micro-fluidic techniques and synthetic biology have testified to the random nature of gene expression and protein abundance in single cells \cite{ref:Yuetal06}, \cite{ref:Friedman10}. Thus a stochastic description of chemical reactions is often mandatory to analyze the behavior of the system.  The dynamics of the system is typically modeled by a continuous-time Markov chain (CTMC) with the state being the number of molecules of each species. \cite{ref:AndKur-11} is a good reference for a review of the tools of Markov processes used in the reaction network systems. 

 Consider a biochemical reaction system consisting of $\Nsp$ species and $\Nreact$ reactions, and let $X(t)$  denote the state of the system at time $t$ in $\Z^\Nsp_+$. If the $k$-th reaction occurs at time $t$, then the system is updated as 
           $X(t) = X(t-) + \change^+_{k}-\change^-_{k},$
where $X(t-)$ denotes the state of the system just before time $t$, and $\change^-_{k}, \change^+_{k} \in \Z^\Nsp_+$ represent the vector of number of molecules consumed and created in one occurrence of reaction $k$, respectively. For convenience, let $\change_k =  \change^+_{k}-\change^-_{k}$. The evolution of the process $X$ is modeled by
$$\PP[X(t+\Delta t) = x +\change_k| X(t) =x] = a_k(x)\Delta t + o(\Delta t).$$
The quantity $a_k$ is usually called the {\em propensity} of the reaction $k$ in the chemical literature, and its expression is often calculated by using the {\em law of mass action} \cite{ref:wilkinson_2006}, \cite{ref:Gillespie2007}. The generator matrix or the $Q$-matrix of the CTMC $X$  is given by
$q_{x,x+\change_k} =a_k(x).$
The CTMC $X$ will have an {\em invariant measure} $\pi$ if $\pi Q\equiv 0$.

\subsection{Site-graphs}

The notion of a  site-graph is a generalization of that of a standard graph. 
A site-graph consists of  nodes and edges; Each node is assigned a set of sites, and the edges are established between two sites of (different) nodes. 
The nodes of a site-graph can be interpreted as protein names, and  sites of a node stand for protein binding domains.
Let $\sites$ denote the set of all the sites in a site-graph, and let $ \powerset(\sites)$ denote the the class of all subsets of $\sites$.

\begin{definition}  
A {\bf site-graph} $\sitegraphdef$ is defined by 
a set of nodes $\nodes$, an interface function $\Sigma:\nodes\ra \powerset(\sites)$, and 
a set of edges 
$\edges\subseteq \{\{(\node,\site),(\node',\site')\}| \node,\node'\in \nodes, \node\neq \node', \site\in \Sigma(\node), \site'\in \Sigma(\node') \}$.
\end{definition}
The function $\Sigma$ in the above definition tracks the sites corresponding to a particular node of a site-graph.

\begin{definition} 
Given a site-graph $\sitegraphdef$, a sequence of edges $(\edge_1,\ldots\edge_k)\in \edges^k$,  
$\edge_i=\{(\node_i,\site_i),(\node_i',\site_i')\}$, such that $\node_i'=\node_{i+1}$ and $s_i'\neq s_{i+1}$ for $i=1,\ldots k-1$, is called a \emph{path} between nodes $\node_1$ and $\node_k$.
If there exists a path between every two nodes $\node,\node'\in \nodes$, a site-graph $\sitegraphdef$ is \emph{connected}. 
\end{definition}

\begin{definition} 
Let $\sitegraph=(\nodes,\Sigma,\edges)$  be a site-graph. A site graph $\sitegraph'$ is a {\bf sub-site-graph} of $\sitegraph$, 
written $\sitegraph'\subgraph\sitegraph$,
if
$\nodes'\subseteq \nodes$, 
for all $\node\in\nodes'$, $\Sigma'(\node)\subseteq \Sigma(\node)$, and 
$\edges'\subseteq \edges$.
\end{definition}

\subsection{Site-graph-rewrite rules}

\begin{definition}\rm
Let $\sitegraphdef$ be a site-graph. 
We introduce two elementary site-graph transformations: adding/deleting an edge. 
\begin{itemize}
\item
$\delta_{ae}(\sitegraph,\edge) =(\nodes_{new}, \Sigma, \edges_{new})$:
$\nodes_{new}=\nodes$, 
$\edges_{new}=\edges\cup \{\edge\}$, 
\item 
$\delta_{de}(\sitegraph,\edge)=(\nodes_{new}, \Sigma, \edges_{new})$:
$\nodes_{new}=\nodes$, 
$\edges_{new}=\edges\setminus \{\edge\}$,
\end{itemize}
The interface function $\Sigma$ is unaltered under any of the above transformations.
Let  $\rhs = (\nodes',\Sigma,\edges')$ be a site-graph derived from $\sitegraphdef$ by a finite number of 
applications of  $\delta_{dn}$, $\delta_{ae}$, $\delta_{de}$. 
Let $\rate\in \R_{\geq 0}$ be a non-negative real number denoting the rate of the transformation.
The triple $(\lhs, \rhs, \rate)$, also denoted by $\lhs \rA{\rate} \rhs$, is called a  {\bf site-graph-rewrite rule}.
\end{definition}

\subsection{Rule-based model}

Suppose that  $\ruleset\equiv \{R_1,\ldots,R_n\}$ is a collection of site-graph rewrite rules such that  for $i=1,\ldots, n$, $R_i\equiv (\sitegraph_i, \sitegraph'_i, \rate_i)$ and $\sitegraph_i = (\nodes_i, \Sigma_i, \edges_i)$. 
From now on, for a given set of rules $\ruleset$, we use the terminology 
\begin{itemize}
\item
the set of \emph{node types} for $\nodes := \cup_i \nodes_i$, 
\item
the set of \emph{edge types} for $\edges := \cup_i \edges_i$,
\item  the \emph{interface function} for $\Sigma: \nodes\ra \powerset(\sites)$, such that for $\node\in\nodes$, $\Sigma(\node): = \cup_i \Sigma_i(\node)$.
\end{itemize}
For each node $v \in \nodes$, 
we will consider $\nins_{\node}$ copies or instances of the node $v$, denoted by $v^1,v^2,\ldots,v^{\nins_{\node}}$. 
Note that, in the Kappa rule-based models, the set of node types and edge types are predefined in the signature of the model; 
Here, it is deduced from the set of rules (a more detailed discussion to the relation with Kappa is given in Section~\ref{Kappa}). 

\begin{definition}
A {\bf reaction mixture} is a site-graph $\usg = (\nodeSet, \fullInterface,\edgeset)$ where
\begin{itemize}
\item $\nodeSet =\{\node^j| \node\in \nodes, j=1,\ldots,n_v\} $;
\item  $\fullInterface(\node^{j})=\Sigma(\node)$;
\item $\edgeset \subset \{\{(\node_1^{i},\site_1),(\node_2^{j},\site_2)\}| \{(\node_1,\site_1),(\node_2,\site_2)\} \in \edges, i=1,\ldots \nins_{\node_1}, j=1,\ldots, \nins_{\node_2}  \}$
\end{itemize}
\end{definition}

\begin{definition}
A \textbf{rule-based model} is a collection of rules $\ruleset$, 
accompanied with the initial reaction mixture $\usg_0$.
\end{definition}

\begin{remark}
By definition, the site-graphs $\lhs_i$ and $G_i'$ occurring in some rule $(\lhs_i,G_i',\rate_i)$, are such that a node $\node\in\nodes$, edge $\edge\in\edges$, but also a site $\site\in\Sigma(\node)$ may be omitted: 
 for some rule $R_i$, we may have a node $\node\in\nodes_i$, such that there exists a site $s\in \Sigma(v)\setminus \Sigma_i(v)$.
The possibility of omitting a site $s\in \Sigma(v)$ from the interface of node $v$ means that the value of site $s$ does not make an influence on the applicability of this rule. 
This is the crucial aspect of reductions of site-graph-rewrite models, because it will help to detect and prove symmetries in the underlying CTMC before considering its full generator matrix.
\end{remark}

\begin{definition}
A rule $(\sitegraph_i,\sitegraph_i',\rate_i)$ is reversible, if there exists a rule $(\sitegraph_j,\sitegraph'_j, \rate_j)$, such that $\sitegraph_i = \sitegraph'_j$ and $\sitegraph_i'={\sitegraph}_j$.
A rule-based model is \textbf{reversible}, if all its rules are reversible.
\end{definition}

Let $\allsg$ be the set of all reaction mixtures which can be reached by finite number of applications of rules from $\ruleset$ to a reaction mixture $\usg_0$.
We will now describe a Markov chain taking values in $\allsg$. The following notion of renaming a site-graph will be used for the formal description.

\begin{definition} \rm
Let $\sitegraphdef$ be a site-graph, $\nodes'$ a set such that $|\nodes'|\geq |\nodes|$ ($|\cdot|$ denotes the set cardinality),  
and $\inst:\nodes\ra\nodes'$ an injective function. Then the {\bf $\eta$-induced node-renamed site-graph}, $\sitegraph^\eta$, is given by $\sitegraph^\inst = (\eta(V), \Sigma^\inst,E^\inst)$, where
$\Sigma^\inst(\inst(\node)) = \Sigma(\node)$ and 
$\edges^\inst =\{\{(\inst(\node_1),\site_1),(\inst(\node_2),\site_2)\} \mid \node_1,\node_2\in\nodes\}$.
\end{definition}

\subsection{The CTMC of a rule-based model}
\begin{figure}
\begin{center}
\includegraphics[width=1\textwidth]{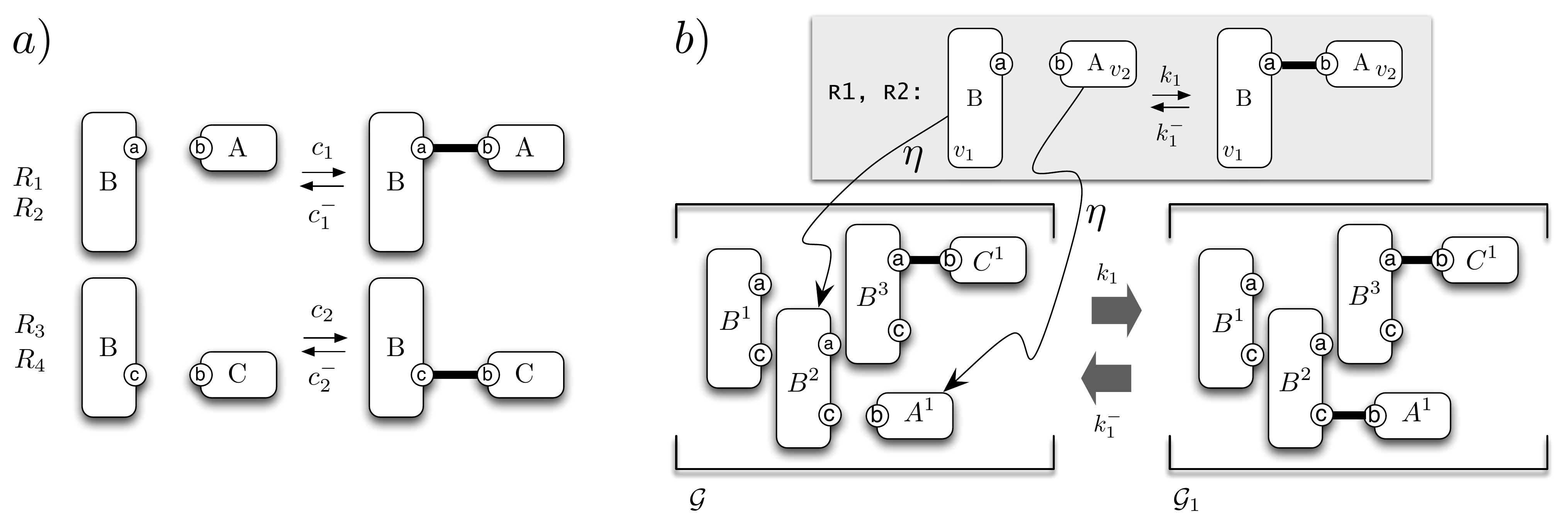}
\caption{Case study 1: Simple scaffold.
a)
The model consists of two reversible rules: a scaffold $B$ has two binding sites, $a$ and $c$, which serve for binding nodes $A$ and $C$, respectively.
b) 
The application of rule $R_1$ to the reaction mixture ${\cal G}$ via node renaming funcion $\eta$ results in a reaction mixture ${\cal G}'$, which is equivalent to ${\cal G}$ except in the sub-site-graph captured by node renaming $\eta$.}
\label{fig:f1}
\end{center}
\end{figure}

Consider a reaction mixture $\usg \in \allsg$, a rule $\rule_i= (\sitegraph_i, \sitegraph'_i, \rate_i) \in \ruleset$. Suppose that $\inst: \nodes\rt\nodeSet$
is a node renaming function such that $\sitegraph_i^\inst \subseteq \usg$. 
This implies that the rule $\rule_i$ can be applied to a part of the reaction mixture $\usg$. 
Let $\usg'_{\inst,i}$ be the unique reaction mixture obtained after the application of the rule $\rule_i$. 
(For a more formal definition of $\usg'_{\inst,i}$ see \cite{lics2010}.)
Note that $\sitegraph_i' \subseteq \usg '$.
Define the transition rate $Q$ by $Q(\usg,\usg'_{\inst,i})=\rate_i$.
More precisely,
\begin{align}
Q(\usg,\usg') = 
\begin{cases}
\rate_i  & \hbox{ if } \usg'= \usg'_{\inst,i} \hbox{ for some } \inst, i
\\
0& \hbox{ if } \usg'= \usg'_{\inst,i} \hbox{ for any } \inst \hbox{ and } i
\\
-\sum_{\usg'\neq \usg}Q(\usg,\usg') & \hbox{ otherwise.}
\end{cases}
\end{align}

Let $\{\proc_t\}$ be a CTMC with state-space $\allsg$ and generator matrix $Q$.

\subsubsection{Case study 1: Simple scaffold.}

Consider a site-graph-rewrite model ${\cal R}\equiv \{R_1, R_2, R_3, R_4\}$, 
depicted in Figure~\ref{fig:f1}a.  
We have that $\nodes=\cup_{i=1}^4 \nodes_i=\{A,B,C\}$, 
$\Sigma(A)=\{b\}$, $\Sigma(B)=\{a,c\}$, $\Sigma(C)=\{b\}$, and
$\edges = \{\{(A,b),(B,a)\},\{(C,b), (B,c)\}\}$.
In Figure~\ref{fig:f1}b, we show the application of rule $R_1$  to the reaction mixture ${\cal G}=({\cal V}, \Sigma, {\cal E})$, such that
${\cal V} = \{A^1,B^1,B^2,B^3, C^1\}$, and 
${\cal E} = \{((B^3,c),(C^1,b))\}$,
$\Sigma(A^1)=\{b\}$, $\Sigma(B^1)=\Sigma(B^2)=\Sigma(B^3)=\{a,c\}$, $\Sigma(C^1)=\{b\}$.


\section{Application} 
\label{sec:application}

This section is devoted to establishing applicability of the results from Section~\ref{sec:cont} to rule-based models. 
Each of the properties - lumpability, invertability and convergence are illustrated on 
three case studies.
For each case study, we first define a trivial uniform aggregation of $\proc_t$,  denoted by $Y_t$, which corresponds to the usual population-based description with mass-action kinetics.
We then show that there exists another uniform aggregation of $\proc_t$, 
denoted by $Z_t$, with much smaller state space.
Finally, since the standard biological analysis are referring to the population-based Markov chain we outline below a method of  retrieving the conditional distribution of $Y_t$ given $Z_t$.
The summary of all considered reductions is given in Table \ref{tab:cs3}.
\label{sec:lumping} 

The following observation establishes an algorithmic criterion for checking \ref{c2} and is obvious from \eqref{eq:delta}. An illustration is given in Figure~\ref{fig:lumping}. 

\begin{lemma}\rm
Let $\partspg = \{\eqclass_1,\ldots, \eqclass_n \}$ be a partitioning of $\allsg$ induced by an equivalence relation $\rel\subseteq \allsg\times \allsg$.
Let $\al_i$ be the uniform probability measure on $\eqclass_i$, that is,
for any $\usg \in \eqclass_i$, 
$\al_i(\usg) = |\eqclass_i|^{-1}$. 
Note that in this case \eqref{defn:cont} reduces to, 
\begin{align}
\label{eq:delta}
\Delta(\eqclass_i,\usg)= 
\frac{|\eqclass_j|}{|\eqclass_i|}
\sum_{\usg_1\in \eqclass_i} \genmat(\usg_1,\usg), \quad \usg \in \eqclass_j. 
\end{align}

Then, the following condition implies \ref{c2}: 
\begin{enumerate}[label={(Cond\arabic*)}, start=3, leftmargin=*, align=right]
\item\label{c3} 
For all $\eqclass_i,\eqclass_j\in \partspg$, for all $\usg,\usg'\in \eqclass_j$, 
there exists a permutation 
of states in $\eqclass_i$, $\sigma:\eqclass_i\rt \eqclass_i$, such that $\genmat(\usg_1,\usg) = \genmat(\sigma(\usg_1),\usg').$
\end{enumerate}
\end{lemma}

\begin{definition}
If the equivalence relation $\rel\subseteq \allsg\times \allsg$  satisfies \ref{c3} and for each $i=1,\ldots,m$, $\alpha_i$ is a uniform probability measure on $\eqclass_i$, then the corresponding Markov chain $\{Y_t\}$ (with generator matrix $\tilde{\genmat}(\eqclass_i,\eqclass_j)\equiv \Delta(\eqclass_i,\usg), \usg\in \eqclass_j$)  is  a \textbf{uniform aggregation} of $\{X_t\}$.
\end{definition}

\begin{figure}
\begin{center}
\includegraphics[width=0.8\textwidth]{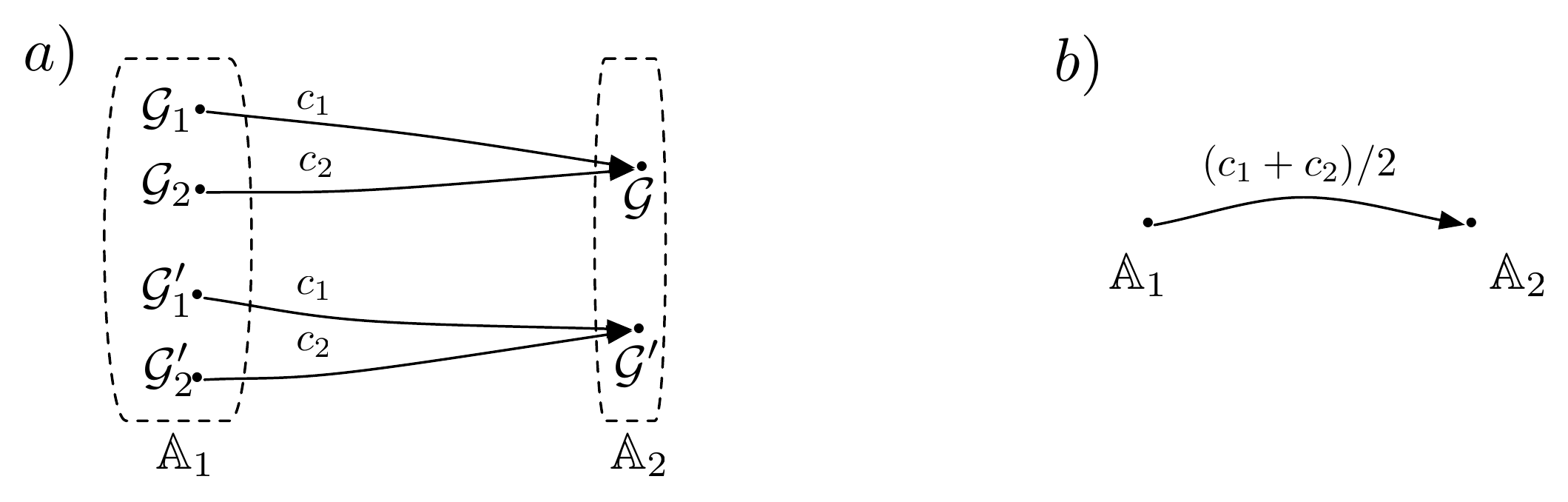} 
\caption{Illustration for testing \ref{c3} and its relation to \ref{c2}:
Let $\eqclass_1 = \{\usg_1,\usg_2,\usg_1',\usg_2'\}$, $\eqclass_2=\{\usg,\usg'\}$.
For $\usg,\usg'\in \eqclass_2$, the permutation 
$\sigma(\usg_1)=\usg_1'$, $\sigma(\usg_1')=\usg_1$, $\sigma(\usg_2)=\usg_2'$ and $\sigma(\usg_2')=\usg_2$ proves that the predecessors of $\usg$ and those of $\usg'$ inside class $\eqclass_1$ are in bijection, that is, \ref{c3} holds.
\ref{c2} follows, since $\tilde{Q}(\eqclass_1,\eqclass_2) = \delta(\eqclass_1,\usg) = \delta(\eqclass_2,\usg')$, which is because
$Q(\usg_1,\usg)+Q(\usg_2,\usg) =Q(\sigma(\usg_1),\usg')+Q(\sigma(\usg_2),\usg') =c_1+c_2$, and
the rate in the aggregated chain is $Q(\eqclass_i,\eqclass_j) = 
\frac{|\eqclass_j|}{|\eqclass_i|}(c_1+c_2)=\frac{1}{2}(c_1+c_2)$. 
}
\label{fig:lumping}
\end{center}
\end{figure}


Let $\relo$ and $\relt$ be two equivalence relations of $\allsg$, such that 
 $\allsg_1 =\{\eqclass_1,\eqclass_2,\ldots\}$ and $\allsg_2=\{\eqcb_1,\eqcb_2,\ldots\}$ are the corresponding sets of equivalence classes. 
Suppose that 
$\relo$ and $\relt$ induce uniform aggregations on $\{X_t\}$, denoted respectively by $\{Y_t\}$ and $\{Z_t\}$.
The property of invertibility allows to evaluate the conditional distributions of $\proc_t$ given $Y_t$, and of $\proc_t$ given $Z_t$. 
However, as mentioned 
oftentimes it is of interest to the modeler to retrieve the conditional distribution of $Y_t$ given $Z_t$. This is possible by the following result. 

\begin{theorem} \rm 
\label{thm:conn}
If $\relo$ is coarser than $\relt$ (that is, $\relo \subseteq \relt$),
then $\allsg_2$ can be obtained by partitioning $\allsg_1$ as follows.  
$$
\eqclass_i\rel \eqclass_j \hbox{ iff there exist } \usg\in\eqclass_i,\usg'\in\eqclass_j, \hbox{ such that } \usg\sim_2\usg'.
$$
Equivalently, 
\begin{align}\label{neweqrel}
\eqclass_i\rel \eqclass_j \mbox{ iff there exists } \eqcb_k \mbox{ such that } \eqclass_i \cup\eqclass_j \subset \eqcb_k.
\end{align}

Assume that $\{Y_t\}$ and $\{Z_t\}$ with generator matrices $\genmat_1$ and $\genmat_2$ are two uniform aggregations of the Markov chain $\{X_t\}$ induced by $(\rel_1,\{\alpha_i\})$ and $(\rel_2,\{\beta_i\})$, where $\alpha_i$ and $\beta_j$ are uniform over $\eqclass_i$ and $\eqcb_j$ respectively.
Define
\begin{align}
\label{eq:alpha}
\al'_j(\eqclass_i) := 
\begin{cases} 
\frac{|\eqclass_i|}{|\eqcb_j|}, & \hbox{ if } \eqclass_i \subseteq \eqcb_j\\
0 & \hbox{, otherwise.}
\end{cases}
\end{align}
Then $\{\alpha'_j\}$ 
satisfies \ref{c2} and hence $\{Y_t\}$ is an aggregation of the Markov chain $\{Z_t\}$. 
\end{theorem}

\begin{proof}
It is trivial to check that $\rel$ defined by \eqref{neweqrel} is a well-defined equivalence relation.

Assume now that $\eqcb_j, \eqcb_{j'}\in\allsg_2$ and $\eqclass_{i'}\subseteq\eqcb_{j'}$. We have to show that
$\Delta(\eqcb_j,\eqclass_{i'})$ is constant for all $\eqclass_{i'}\subseteq\eqcb_{j'}. $ Toward this end notice that
\begin{align*}
\Delta(\eqcb_j,\eqclass_{i'})=&
\frac{\sum_{i:\eqclass_i\subseteq\eqcb_{j}}\alpha_j(\eqclass_i)Q_1(\eqclass_i,\eqclass_{i'})}{\alpha_j(\eqclass_{i'})} =
\frac{\sum_{i:\eqclass_i\subseteq\eqcb_j}|\eqclass_i|/|\eqcb_j|Q_1(\eqclass_i,\eqclass_{i'})}{|\eqclass_{i'}|/|\eqcb_{j'}|} \\
=& \frac{\sum_{\eqclass_i\subseteq\eqcb_j}|\eqclass_i|/|\eqcb_j| \sum_{\usg'\in\eqclass_i} Q(\usg',\usg) |\eqclass_{i'}|/|\eqclass_i| }{|\eqclass_{i'}|/|\eqcb_{j'}|}, \quad \hbox{ for some } \usg \in \eqclass_{i'}\\ 
=& 
\frac{\sum_{\eqclass_i\subseteq\eqcb_j} \sum_{\usg'\in\eqclass_i} Q(\usg',\usg) |\eqclass_i|/|\eqcb_j| |\eqclass_i'|/|\eqclass_i| }{|\eqclass_{i'}|/|\eqcb_{j'}|}\\
=& 
{\sum_{\usg'\in\eqcb_j} Q(\usg', \usg) |\eqcb_{j'}|/|\eqcb_j|} \\
=& Q_2(\eqcb_j,\eqcb_{j'}). 
 \end{align*}

Here the third and the last equalities are  because by the assumption  $\{Y_t\}$ and $\{Z_t\}$ are  uniform aggregations of $\{X_t\}$. 
\end{proof}
\subsection{Case study 1: Simple scaffold (continued)}

\begin{figure}
\begin{center}
\includegraphics[width=1\textwidth]{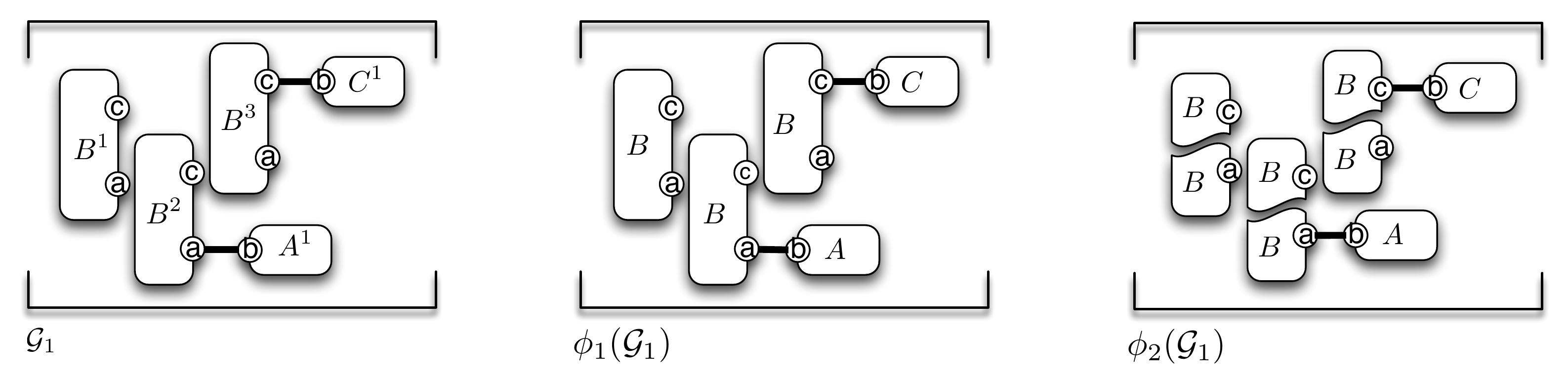}
\caption{Case study 1: Simple scaffold. 
A reaction mixture ${\cal G}_1$(left) and its graphical representation in aggregation $\phi_1$ (center) or $\phi_2$ (right). 
}
\label{fig:ex12}
\end{center}
\end{figure}
The simple scaffold example serves as an illustrative case study which demonstrates all the introduced concepts in detail.

\begin{description}
\item[Species.]
A molecular species is a class of connected reaction mixtures which are isomorphic up to renaming of the nodes of same type. 
We here omit a formal definition of species, since it is not necessary for conveying the arguments. 
In the scaffold example, all species can be categorized into six types: 
$(A)$-- a free node of type $A$, 
$(B)$-- a free node of type $B$,
$(C)$-- a free node of type $C$,
$(AB)$-- a node of type  $B$ that is bound to node of a type $A$, and is not bound to a node of type $C$, 
$(BC)$-- a node of type $B$ that is bound to a node of type $C$, and is not bound to a node of type $A$, and
$(ABC)$-- a node of type $B$ that is bound to a node of type $A$, and is also bound to a node of type $C$.
All reaction mixtures ${\usg}$, which count the same number of each of the species correspond to the same population-based state. 
The population-based encoding of the state space is captured by the function
$\funo:\allsg\ra \N^3$, such that
$
\funo(\usg)=(\mAB,\mBC,\mABC),
$  
if $\usg$ has 
$\mAB$ sub-site-graphs of type $(AB)$, 
$\mBC$ sub-site-graphs of type $(BC)$, and 
$\mABC$ sub-site-graphs of type $({ABC})$.
Note that, given the value $\funo(\usg)$, the number of sub-site-graphs of type $(A)$, $(B)$ and $(C)$ in $\usg$ is also known, since the total number of nodes of each type is conserved. 
Two reaction mixtures $\usg$ and ${\usg'}$  are aggregated by relation 
$\relo\subseteq \allsg\times \allsg$ if they have the same value of function $\phi_1$:
\[\usg\relo{\usg'} \hbox{ iff } 
\funo(\usg) = \funo({\usg'}).\]
For example, in Figure~\ref{fig:cs1}, $\funo(\usg_1)\neq \funo(\usg_2)$.

The aggregated CTMC, $\{Y_t\}$, takes values in $\N^3$, and it is exactly the standard population-based model description with mass-action kinetics.

\item[Fragments.]
The sites $a$ and $c$ of nodes of type $B$ are updated without testing each-other. 
As formally shown later in Lemma \ref{lem:lem}, any two states which have the same number of free sites $c$ and free sites $a$ are not distinguishable by the system's dynamics. 
As a consequence, the following lumping is also applicable:
let $\funt:\allsg\ra \N^2$ be such that
$\funt(\usg)=(\mABq,\mBCq)$, 
if $\usg\in\allsg$ has $\mABq$  nodes $B$ bound to $A$ and
$\mBCq$ nodes $B$ that bound to $C$.
The two states $\usg$ and $\usg'$ are aggregated by relation $\relt\subseteq \allsg\times \allsg$ if they have the same value of function $\phi_2$:
\[\usg\relt \usg' \hbox{ iff } 
\funt(\usg) = \funt(\usg').\]

For example, in Fig.~\ref{fig:cs1}, $\funt(\usg_1)=\funt(\usg_2)$.
The aggregation of $\{X_t\}$ by $\funt$ results in a CTMC $\{Z_t\}$, which takes values in $\N^2$, and it therefore provides a better reduction than the standard population-based model description. 
A way to visualize the states of CTMC's $\{X_t\}$, $\{Y_t\}$ and $\{Z_t\}$ is shown in Fig.~\ref{fig:ex12}.
\end{description}

\begin{lemma} 
\label{lem:lem}
\rm
Both relations $\relo$ and $\relt$ induce uniform aggregations of $\{X_t\}$.
Moreover, 
$\relo\subseteq \relt$, that is, $\relt$ is coarser than $\relo$.
\end{lemma}

\begin{proof} 
Consider lumping by $\relt$.
Let $\usg_1,\usg_2$ be two reaction mixtures such that $\usg_1\relt \usg_2$, and let $\funt(\usg_1)=\funt(\usg_2)=(\mABq,\mBCq)$. 
If $\usg_1,\usg_2\in\eqcb_j$, by Theorem \ref{thm:conn}, it is enough to show that for any $\eqcb_i\in\allsg_2$, and any $\usg\in\eqcb_i$, there is a permutation 
$\sigma:\eqcb_i\ra\eqcb_i$, such that
$Q(\usg,\usg_1)=Q(\sigma(\usg),\usg_2)$.
Choose some $\eqcb_i\in\allsg_2$ and $\usg\in\eqcb_i$.
Then, $\funt(\usg) \in\{(\mABq-1,\mBCq),(\mABq+1,\mBCq),(\mABq,\mBCq+1),(\mABq,\mBCq-1)\}$.
We analyze the case  $\funt(\usg) =(\mABq-1,\mBCq)$; the other three cases are analogous.

Let $\usg_1 = (\nodeSet, \fullInterface,\edgeset_1)$ and $\usg_2 = (\nodeSet, \fullInterface,\edgeset_2)$. 
Since $\funt(\usg_1)=\funt(\usg_2)$, there exists a bijective renaming function $\eta:\nodeSet\ra\nodeSet$,
such that $\usg_2=\usg_1^{\eta}$, that is, $\usg_2$ is $\eta$-induced node-renamed site-graph $\usg_1$. 
It is easy to inspect that $Q(\usg,\usg_1)=c_1$ if and only if $Q(\usg^{\eta},\usg_2)=c_1$. 
So - the bijection over the reaction mixtures aggregated to $\eqcb_i$ is the one induced by renaming $\eta$. 


For showing that $\relt$ is coarser than $\relo$, it is enough to observe that the map $\phi:\N^3\ra \N^2$ defined by $\phi(m_{AB},m_{BC}, \mABC)=(\mAB+\mABC, \mBC+\mABC)$ is such that 
$\phi_2=\phi\circ \phi_1$.
\end{proof}





Consequently to  Lemma \ref{lem:lem}, Theorem \ref{thm:conn} applies.
Then, the process $\{Z_t\}$ is also lumpable with respect to $\{Y_t\}$, and Theorem \ref{thmct_sdinv} applies. 
Imagine that it is possible to experimentally 
synthesize only the complexes of type $(AB)$ and of $(BC)$, but not a complex of type $(A)$, $(B)$, $(C)$ or $({ABC})$.
Then, the initial distribution does not respect $\al_i$, as soon as $\nins_A\geq 1$, $\nins_B\geq 2$, $\nins_C\geq 1$.
However, since each reversible rule-based model trivially has an irreducible CTMC, the Theorem \ref{thmct_conv} holds.

A concrete example is demonstrated in Figure \ref{fig:cs1}.
The details for the calculation for Table \ref{tab:cs3}, de-aggregation, 
as well the discussion for $n_A=n_C=1$, $n_B=2$ can be found in the Appendix. 


\begin{figure}
\begin{center}
\includegraphics[width=1.\textwidth]{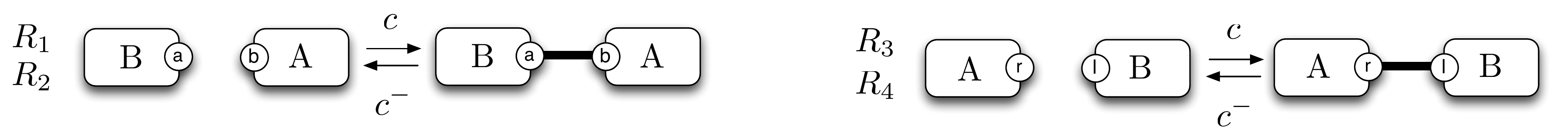} 
\caption{Case study 2: Two-sided polymerization.
} 
\label{fig:ex12b}
\end{center}
\end{figure}

\subsection{Case study 2: Two-sided polymerization}

\begin{figure}
\begin{centering}
\includegraphics[width=.9\linewidth]{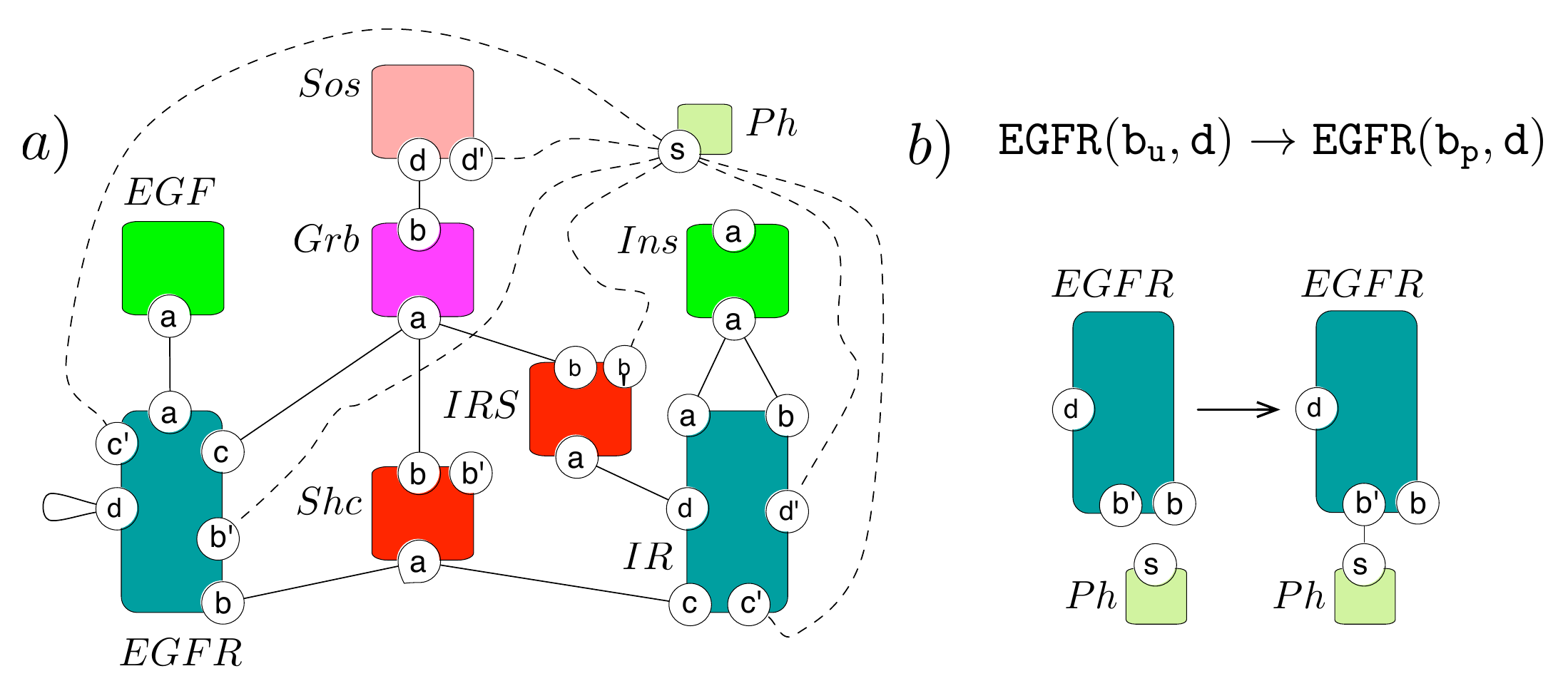}
\caption{
a) 
Summary of interactions between nodes in case study 3. 
The dotted lines represent phosphorylation, and solid lines denote standard bindings. 
The self-loop at the site $d$ of node $EGFR$ means that it can bind to another node $EGFR$, i.e. receptor dimerization. 
b)
An example of a Kappa rule, and a corresponding site-graph rewrite rule.
}
\label{fig:cs3}
\end{centering}
\end{figure}

The two-sided polymerization case study illustrates the drastic advantage of using the fragment-based CTMC, because it shows to have exponentially smaller state space than the species-based CTMC. 

Consider a site-graph-rewrite model ${\cal R}$ 
depicted in Fig.~\ref{fig:ex12}b: proteins $A$ and $B$ can polymerize by forming bonds of two kinds: between site $b$ of protein $A$ and site $a$ of protein $B$, or between site $r$ of protein $A$ and site $l$ of protein $B$. 
Assume  that there are $\nA$ nodes of type $A$ and $\nB$ nodes of type $B$. 
Let $\allsg$ be the set of all reaction mixtures.
All connected site-graphs  occurring in a reaction mixture can be categorized into two types: 
\emph{chains} and \emph{rings}.  \emph{Chains} are the connected site-graphs having two free sites, and \emph{rings} are those having no free sites.
We say that a chain or a ring is of length $i$ if it has $i$ bonds in total.  Chains can be classified into four different kinds, depending on which sites are free.

\begin{description}
\item[Species.]
Let $\funo:\allsg\ra \N^{5m}$ be such that
\[\funo(\usg)=(x_{11},\ldots, x_{1m},x_{21},\ldots,x_{2m},x_{31},\ldots,x_{3m},
x_{41},\ldots,x_{4m},x_{51},\ldots,x_{5m}),\] 
if $\usg\in\allsg$ has 
\begin{itemize}
\item 
$x_{1i}$  chains of type $(A..B)_i$, that is, of length $2i-1$, with free sites $b$ and $a$, 
\item 
$x_{2i}$  chains of type $(B..A)_i$, that is, of length $2i-1$, with free sites $l$ and $r$, 
\item 
$x_{3i}$  chains of type $(A..A)_i$, that is, of length $2i$, with free sites $b$ and $a$, 
\item 
$x_{4i}$  chains of type $(B..B)_i$, that is, of length $2i$, with free sites $l$ and $r$, 
\item 
$x_{5i}$  rings of type $(.A..B.)_i$, that is, of length $2i$.
\end{itemize}
The two states $\usg$ and $\tilde{\usg}$ are aggregated by the equivalence relation 
$
\relo\subseteq \statesp\times \statesp \hbox{ if } 
\funo(\usg) = \funo(\tilde{\usg}).
$

\item[Fragments.]
Let $\funt:\allsg\ra \N^2$ be such that
$\funt(\usg)=(\nbonds_{rl},\nbonds_{ba})$, 
if $\usg\in\allsg$ has $\nbonds_{rl}$ bonds between sites $r$ and $l$, and $\nbonds_{ba}$ bonds between sites $b$ and $a$.
The two states $\usg$ and ${\usg'}$ are aggregated by the equivalence relation 
$
\relt\in \allsg\times \allsg \hbox{ if } 
\funt(\usg) = \funt({\usg'}).
$

Alternatively, since the rates of forming and releasing bonds do not depend on the type of the bond, 
let $\funth:\allsg\ra \N$ be such that
$\funth(\usg)=\nbonds$,
if $\usg\in\allsg$ has in total $\nbonds$ bonds.
The two states $\usg$ and ${\usg'}$ be aggregated by equivalence relation 
$\relth\in \allsg\times \allsg \hbox{ if } 
\funth(\usg) = \funth({\usg'}).$
\end{description}

A concrete example is demonstrated in Figure \ref{fig:cs2}.
The details for the calculation for Table \ref{tab:cs3}, and on de-aggregation can be found in the Appendix.

\begin{figure}
\begin{centering}
\includegraphics[width=1\linewidth]{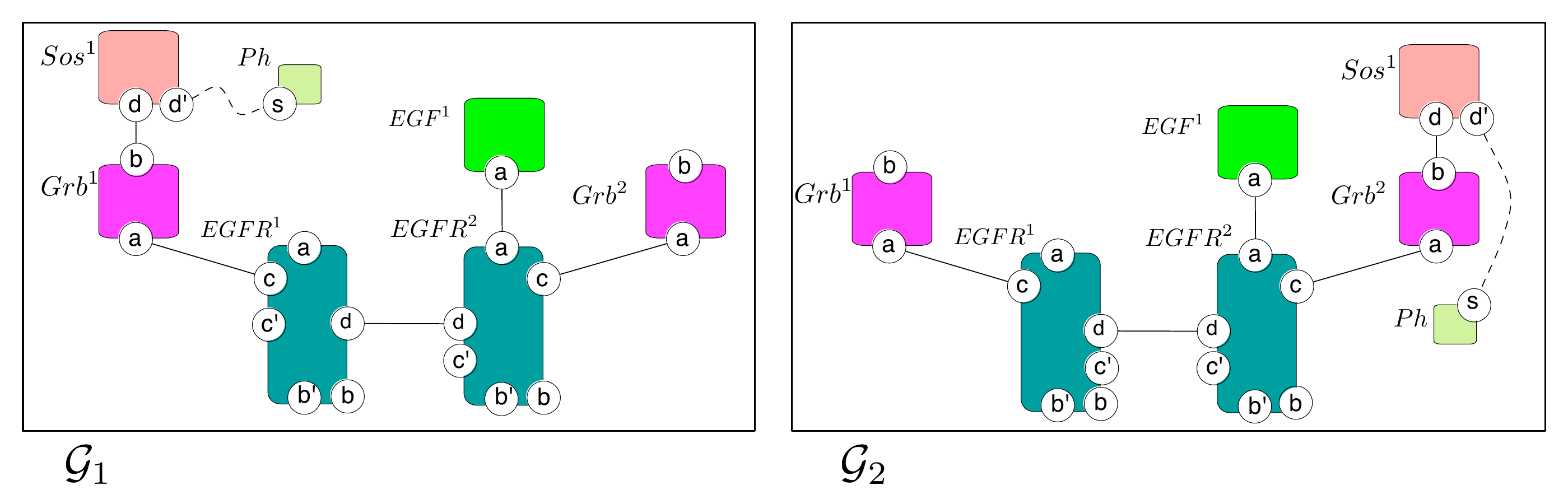}
\caption{
Case study 3: reaction mixtures $\usg_1$, $\usg_2$, such that they are aggregated in the fragment description -- both states contain one protein $Grb$ that is free on site $b$, one protein $Grb$ that is bound to a site $d$ of protein $Sos$, and one species containing a dimer of $EGFR$ proteins, such that each $EGFR$ protein is bound to one $Grb$ protein, and one of them is bound to an $EGF$ protein.
Let $\usg_1\in\eqclass_1\subseteq \eqcb_1$ and  $\usg_2\in\eqclass_2\subseteq \eqcb_1$.
Then, by
Theorem \ref{thmct_sdinv}, we have that 
$\PP(Z_t=\eqcb_1)=\PP(Y_t\in \{\eqclass_1,\eqclass_2\})$ (lumpability), and
$\PP(Y_t=\eqclass_1)=0.5\PP(Z_t=\eqcb_1)$ whenever $\PP(Y_0=\eqclass_1)=\PP(Y_0=\eqclass_2)$ (invertability).
Moreover, by Theorem \ref{thmct_conv},
$\PP(Y_t=\eqclass_1)\ra 0.5\PP(Z_t=\eqcb_1)$, when $t\ra \infty$ (convergence). 
}
\label{fig:cs3a}
\end{centering}
\end{figure}
\subsection{Case study 3: EGF/insulin pathway}

We take a model of the network of interplay between insulin and epidermal growth factor (EGF) signaling in  mammalian cells from literature \cite{conzelmann2008}. 
The original model suffers from the huge number of feasible multi-protein species and the high complexity of the related reaction networks.
It contains $42956$ reactions and $2768$ different molecular species, i.e. connected reaction mixtures which differ up to node identifiers.
The reactions can be translated into a Kappa model of only $38$ transition rules. 

The bases for the framework of site-graph-rewrite models used in this paper is a rule-based modeling language Kappa \cite{Kappa}.
A Kappa rule and an example of the corresponding site-graph-rewrite rule are shown in Figure \ref{fig:cs3}b. 
The general differences to Kappa are detailed in Section \ref{Kappa}. 
In Figure \ref{fig:cs3}a, we show the summary of protein interactions for this model, adapted to the site-graph-rewrite formalism used in this paper. 
Due to the independence between the sites $a$ and $b$ of protein $Grb$, it was proven in \cite{lumpability}, that it is enough to track the copy number of $609$ partially defined complexes, that are named \emph{fragments}. 
Thus, the dimension of the state vector in the reduced system is $609$, instead of $2768$ in the concrete system.

%


\begin{description}
\item[Species.]
Two reaction mixtures $\usg$ and $\tilde{\usg}$  are aggregated by relation 
$\relo\subseteq \allsg\times \allsg$ if they contain the same number of molecular species.

\item[Fragments.]
Let a fragment be a part of a  molecular species that 
either does not contain protein $Grb$, 
or it contains only a site $a$ of protein $Grb$, 
or it contains only a site $b$ of protein $Grb$.
Two reaction mixtures $\usg$ and ${\usg'}$  are aggregated by relation 
$\relt\subseteq \allsg\times \allsg$ if they contain the same number 
of fragments.
A concrete example is demonstrated in Figure \ref{fig:cs3a}.
\end{description}



\section{Conclusion}
In this paper, we have studied model reduction for a Markov chain using aggregation techniques.
We provided a sufficient condition for defining a CTMC over the aggregates, a \emph{lumpable} reduction of the original one. 
Moreover, we characterized sufficient conditions for \emph{invertability}, that is, when the measure over the original process can be recovered from that of the aggregated one.
We also established  \emph{convergence} properties of the aggregated process and showed how lumpability and invertability depend on the initial distribution. 
Three case studies demonstrated the usefulness of the techniques discussed in the paper.

\begin{table}
\begin{centering}
\begin{tabular}{|c|c|}
\hline
\begin{tabular}{|c|}
\hline\\
\hline
Simple scaffold \\
($3$ node types)
\\
\hline
Polymerization \\
($2$ node types)
\\
\\
\hline
EGF/insulin \\
($8$ node types)
\\
\hline
\end{tabular}
& 
\begin{tabular}{|c|c|c|c|}
\hline 
\textbf{lumping}& \# \textbf{rules} & \textbf{dim.} & \textbf{estimated \# of states}\\
\hline
species & $8$ & $3$ & $(n+1)(n+2)(n+3)/6$\\
fragment& $4$ & $2$ & $(n+1)^2$\\
\hline
species & - & $n$ & $>3P(n)$ \\
fragment& $4$ & $2$ & $(n+1)^2$\\
fragment 2 & 2 & 1 & 2n+1 \\
\hline
species & $ 42956$ & $2768$ & - \\
fragment& $38$ & $609$ & -\\
\hline
\end{tabular}
\\
\hline
\end{tabular}
\caption{
Summary of the reduction for the presented case studies.
In case study 1, 
for $n_A=n_B=n_C=n$, the number of states is reduced from $O(n^3)$ to $O(n^2)$.
The number of partitions of $n$ is denoted by $P(n)\approx \frac{1}{4n\sqrt{3}} e^{\pi\sqrt{\frac{2n}{3}}}$ \cite{ramanujan18}.
In case study 2, 
for $n_A=n_B=n$,  there is an exponential reduction in the number of states from standard to the aggregated CTMC.
In case study 3 (a crosstalk between the epidermal growth factor, EGF, and insulin pathway), the dimension of the state vector is reduced from $2768$ to $609$, and we did not estimate the size of the state space.
}
\label{tab:cs3}
\end{centering}
\end{table}

\begin{acknowledgements}
A. Ganguly and H. Koeppl acknowledge the support from the Swiss National Science Foundation, grant number PP00P2 128503/1. T. Petrov is supported by SystemsX.ch - the Swiss Inititative for Systems Biology.
\end{acknowledgements}



\bibliographystyle{plain}

\begin{thebibliography}{10}

\bibitem{ref:AndKur-11}
D.~F. Anderson and T.~G. Kurtz.
\newblock Continuous time markov chain models for chemical reaction networks.
\newblock In H.~Koeppl, G.~Setti, M.~di~Bernardo, and D.~Densmore, editors,
  {\em Design and Analysis of Biomolecular Circuits}. Springer-Verlag, 2011.

\bibitem{bli:06}
Michael~L. Blinov, James~R. Faeder, Byron Goldstein, and William~S. Hlavacek.
\newblock A network model of early events in epidermal growth factor receptor
  signaling that accounts for combinatorial complexity.
\newblock {\em BioSystems}, 83:136--151, January 2006.

\bibitem{buchholz_lump}
Peter Buchholz.
\newblock Exact and ordinary lumpability in finite {Markov} chains.
\newblock {\em Journal of Applied Probability}, 31, no1:59--75, 1994.

\bibitem{buchholz_bisimulation}
Peter Buchholz.
\newblock Bisimulation relations for weighted automata.
\newblock {\em Theoretical Computer Science}, Volume 393, Issue 1-3:109--123,
  2008.

\bibitem{conzelmann2008}
Holger Conzelmann, Dirk Fey, and Ernst~D. Gilles.
\newblock Exact model reduction of combinatorial reaction networks.
\newblock {\em BMC Systems Biology}, 2(78):342--351, 2008.

\bibitem{lics2010}
Vincent Danos, Jerome Feret, Walter Fontana, Russell Harmer, and Jean Krivine.
\newblock Abstracting the differential semantics of rule-based models: Exact
  and automated model reduction.
\newblock {\em Symposium on Logic in Computer Science}, 0:362--381, 2010.

\bibitem{Kappa}
Vincent Danos and Cosimo Laneve.
\newblock Core formal molecular biology.
\newblock {\em Theoretical Computer Science}, 325:69--110, 2003.

\bibitem{lumpability}
Jerome Feret, Thomas Henzinger, Heinz Koeppl, and Tatjana Petrov.
\newblock Lumpability abstractions of rule-based systems.
\newblock {\em Theoretical Computer Science}, 431(0):137 -- 164, 2012.

\bibitem{ref:Friedman10}
Nir Friedman, Long Cai, and X.~Sunney Xie.
\newblock Stochasticity in gene expression as observed by single-molecule
  experiments in live cells.
\newblock {\em Israel Journal of Chemistry}, 49:333--342, 2010.

\bibitem{ref:Gillespie2007}
Daniel~T. Gillespie.
\newblock Stochastic simulation of chemical kinetics.
\newblock {\em Annual Review of Physical Chemistry}, 58(1):35--55, 2007.

\bibitem{ref:Gup95}
P.~Guptasarma.
\newblock {Does replication-induced transcription regulate synthesis of the
  myriad low copy number proteins of Escherichia coli?}
\newblock {\em BioEssays : news and reviews in molecular, cellular and
  developmental biology}, 17(11):987--997, November 1995.

\bibitem{ramanujan18}
Hardy and Ramanujan.
\newblock Asymptotic formula in combinatory analysis.
\newblock {\em Proceedings of the London Mathematical Society},
  S2-17(1):75--115, 1918.

\bibitem{ref:Her-LerLas-03}
O.~Hern{\'a}ndez-Lerma and J.-B. Lasserre.
\newblock {\em Markov {C}hains and {I}nvariant {P}robabilities}, volume 211 of
  {\em Progress in Mathematics}.
\newblock Birkh\"auser Verlag, Basel, 2003.

\bibitem{hlavacekws_2003}
William~S. Hlavacek, James~R. Faeder, Michael~L. Blinov, Alan~S. Perelson, and
  Byron Goldstein.
\newblock The complexity of complexes in signal transduction.
\newblock {\em Biotechnol. Bio-eng.}, 84:783--794, 2005.

\bibitem{ref:Kei87}
Joel Keizer.
\newblock {\em {Statistical Thermodynamics of Nonequilibrium Processes}}.
\newblock Springer, 1 edition, July 1987.

\bibitem{KS60:lump}
John Kemeny and James~L. Snell.
\newblock {\em Finite {M}arkov Chains}.
\newblock Van Nostrand, 1960.

\bibitem{JLedoux95}
James Ledoux.
\newblock On weak lumpability of denumerable {M}arkov chains.
\newblock {\em Statist. Probab. Lett.}, 25(4):329--339, 1995.

\bibitem{sasb2011}
Tatjana Petrov, Arnab Ganguly, and Heinz Koeppl.
\newblock Model decomposition and stochastic fragments.
\newblock {\em Electronic Notes in Theoretical Computer Science}, 284(0):105 --
  124, 2012.

\bibitem{weakLumpCTMC}
Gerardo Rubino and Bruno Sericola.
\newblock A finite characterization of weak lumpable {Markov Processes}. part
  {II}: The continuous time case.
\newblock {\em Stochastic processes and their applications}, vol. 38,
  no2:195--204, 1991.

\bibitem{weakLumpDTMC}
Gerardo Rubino and Bruno Sericola.
\newblock A finite characterization of weak lumpable {Markov} processes. part
  {I}: The discrete time case.
\newblock {\em Stochastic processes and their applications}, vol. 45, no
  1:115--125, 1993.

\bibitem{Sokolova03onrelational}
Ana Sokolova and Erik P.~de Vink.
\newblock On relational properties of lumpability.
\newblock In {\em Proceedings of the 4th PROGRESS}, 2003.

\bibitem{lumpCommutCTMCExact}
J.~P. Tian and D.~Kannan.
\newblock Lumpability and commutativity of {Markov} processes.
\newblock {\em Stochastic analysis and Applications}, 24, no3:685--702, 2006.

\bibitem{walshct_2006}
Christopher~T. Walsh.
\newblock {\em Posttranslation Modification of Proteins: Expanding Nature's
  Inventory}.
\newblock Roberts and Co. Publisher, 2006.

\bibitem{ref:wilkinson_2006}
D.~J. Wilkinson.
\newblock {\em Stochastic Modelling for Systems Biology}.
\newblock Chapman \& Hall, 2006.

\bibitem{ref:Yuetal06}
J.~Yu, J.~Xiao, X.~Ren, K.~Lao, and X.~S. Xie.
\newblock Probing gene expression in live cells, one protein molecule at a
  time.
\newblock {\em Science}, 311(5767):1600--3, 2006.

\end{thebibliography}

%
%

\section*{Appendix}

\subsection{De-aggregation: simple scaffold}

\begin{figure} 
\begin{center} 
\includegraphics[width=1.\textwidth]{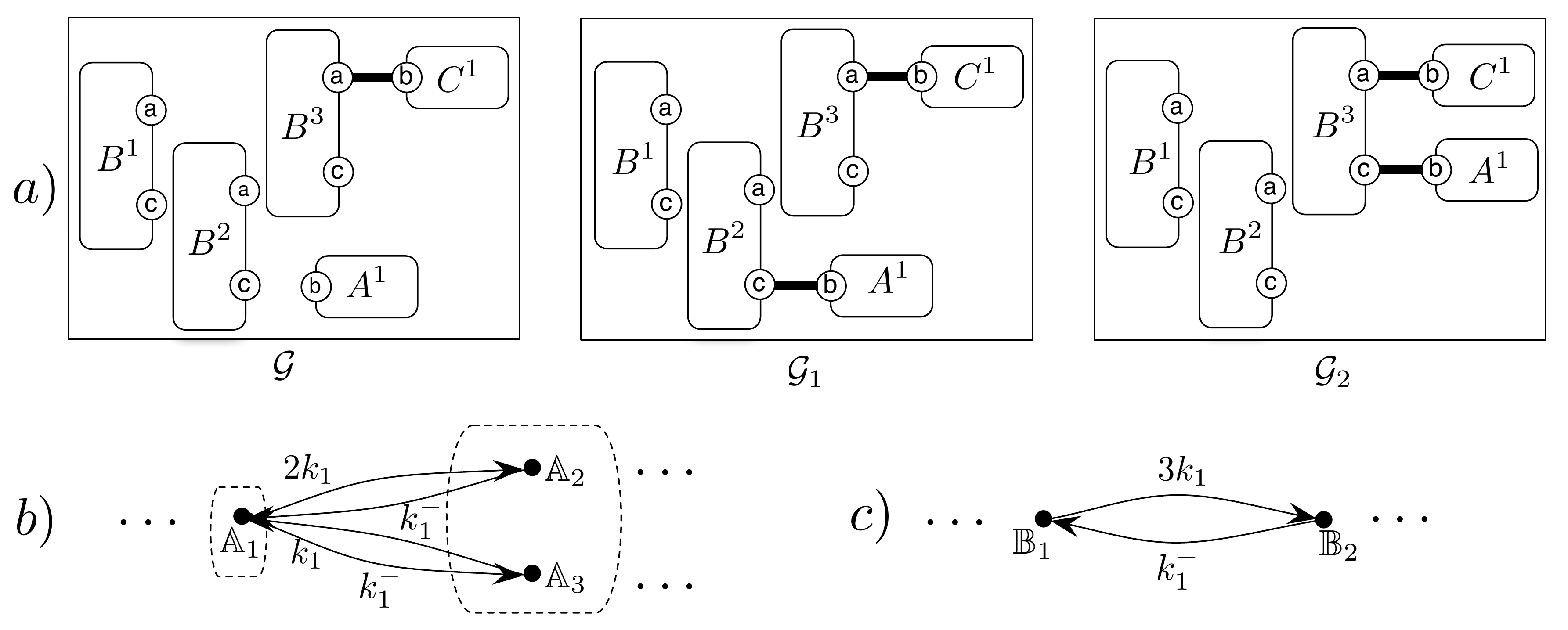} 
\caption{
Interpreting the case study 1 (simple scaffold). 
a) 
examples of reaction mixtures -- $\usg$, $\usg_1$ and $\usg_2$; 
b) 
a part of the CTMC $Y_t\in \{\eqclass_1,\eqclass_2,\eqclass_3,\ldots\}$,
such that $\usg\in \eqclass_1$, $\usg_1\in \eqclass_2$, $\usg_2\in \eqclass_3$; 
c) 
a part of the CTMC $Z_t\in\{\eqcb_1,\eqcb_2,\ldots\}$,
such that $\usg\in \eqcb_1$, $\usg_1,\usg_2\in \eqcb_2$.
The state $\eqcb_2$ is lumping of states $\eqclass_2$ and $\eqclass_3$.
Then, by
Theorem \ref{thmct_sdinv}, we have that 
$\PP(Z_t=\eqcb_2)=\PP(Y_t\in\{\eqclass_1,\eqclass_2\})$ (lumpability), and 
$\PP(Y_t=\eqclass_1)=2/3\PP(Z_t=\eqcb_2)$, 
$\PP(Y_t=\eqclass_2)=1/3\PP(Z_t=\eqcb_2)$, whenever 
$\PP(Y_0=\eqclass_1)=2\PP(Y_0=\eqclass_2)$ (invertability).
Moreover, by Theorem \ref{thmct_conv},
$\PP(Y_t=\eqclass_1)\ra 2\PP(Y_t=\eqclass_2)$, as $t\ra \infty$ (convergence).
}
\label{fig:cs1} 
\end{center}
\end{figure}

Assume that $\usg\in \allsg$ is such that 
$\funo(\usg)=(m_{AB},m_{BC},m_{ABC})$.
Let $m_A:=n_A-\mAB - \mABC$, $m_B:=n_A-\mAB-\mBC-\mABC$ and $m_C:=n_C-\mBC-\mABC$. 
If $\usg\in \eqclass_i$, then $\al_{1i}(\usg)=|\eqclass_i|^{-1}$, where
\begin{align}
\label{eq:alphao}
|\eqclass_i| =\frac{n_A!n_B!n_C!}{m_{AB}!m_{BC}!m_{ABC}!m_A!m_B!m_C!}. 
\end{align}
The explanation is as follows. The $m_A$ free nodes of type $A$, $m_B$ free  nodes of type $B$ and $m_C$ free nodes of type $C$ can be chosen in ${n_A\choose m_A}{n_B\choose m_B}{n_C\choose m_C}$ possible ways. 
Among the remaining nodes, $\mAB$ nodes of type $A$ and $\mAB$ nodes of type $B$ can be chosen in ${n_A-m_A\choose\mAB}{n_B-m_B\choose\mAB}$ ways. 
There are $\mAB!$ different ways to establish bonds between $\mAB$ identified nodes $A$ and $\mAB$ identified nodes $B$.
In the same way, we choose $\mBC$ complexes of type $(BC)$ among the $n_B-m_B-\mAB$ nodes of type $A$, and $n_C-m_C$ nodes of type $B$. 
Finally, there is exactly one way to choose $\mABC$ complexes of type $(ABC)$ among the $n_A-m_A-\mAB$, $n_B-m_B-\mAB-\mBC$ and $n_C-m_C-\mBC$ nodes of type $A$, $B$ and $C$ respectively. 
Connecting the bonds can be done in $(\mABC!)^2$ different ways (for each node $B^j$, there are exactly $\mABC!$ ways to choose the $A^i$ and $\mABC!$ ways to choose $C^k$).
The final expression follows. 

Moreover, if $\funt(\usg)=(\mABq,\mBCq)$ and $\usg\in \eqcb_j$, then $\al_{2j}(\usg) = |\eqcb_j|^{-1}$, where
\begin{align}
\label{eq:alphat}
|\eqcb_j|=
{n_A\choose \mABq}{n_B\choose \mABq}\mABq!
{n_C\choose \mBCq}{n_B\choose \mBCq}\mBCq!.
\end{align}

We first choose the $\mABq$ nodes of type $A$ and $\mABq$ nodes of type $B$; There are $\mABq!$ different ways to establish the bonds; In total, it makes ${n_A\choose \mABq}{n_B\choose \mABq}\mABq!$ choices.
Independently, the $\mBCq$ bonds between $B$ and $C$ can be chosen
in ${n_B\choose \mBCq}{n_C\choose \mBCq} \mBCq!$ ways.

\subsection{De-aggregation: two-sided polymerization}

\begin{figure}
\begin{center}
\includegraphics[width=1.\textwidth]{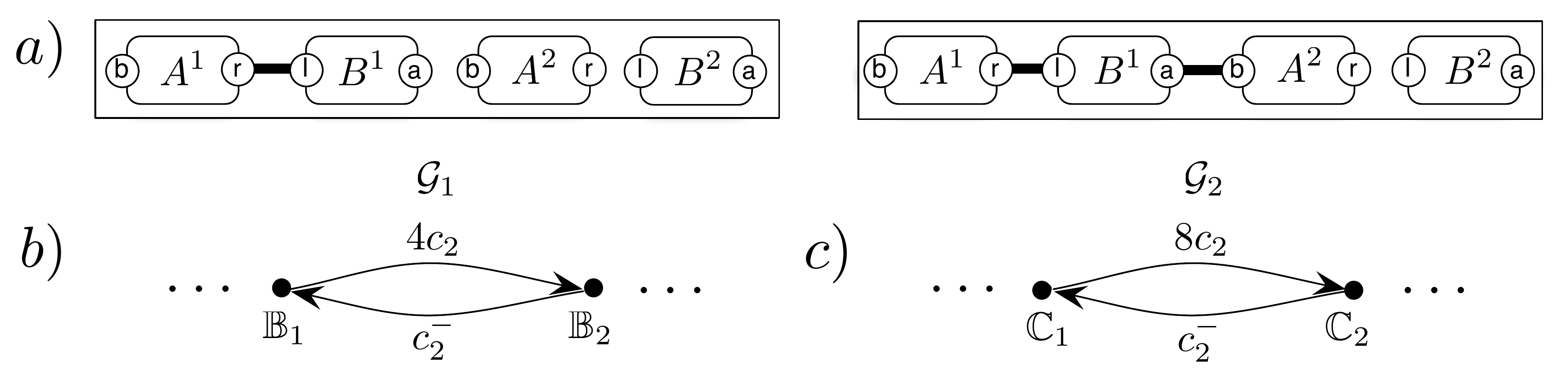} 
\caption{ 
Case study 2: two-sided polymerization.
a) 
examples of reaction mixtures; 
b) 
a part of the CTMC $Z_t\in\{\eqcb_1,\eqcb_2,\ldots\}$,
such that $\usg_1\in \eqcb_1$, $\usg_2\in \eqcb_2$,
c) 
a part of the CTMC $Z'_t\in\{{\mathbb C}_1,{\mathbb C}_2,\ldots\}$, 
such that $\usg_1\in {\mathbb C}_1$, $\usg_2\in {\mathbb C}_2$. 
} 
\label{fig:cs2}
\end{center}
\end{figure}

Assume that $s$ is a site-graph such that
$$\funo(s)=
(x_{11},\ldots, x_{1m},
x_{21},\ldots,x_{2m},
x_{31},\ldots,x_{3m},
x_{41},\ldots,x_{4m},
x_{51},\ldots,x_{5m}).
$$
We do not give the analytic expression for $\al_{1i}(s)$. 
For computing it, it is enough to use the following:
\begin{itemize}
\item 
choosing a chain of type $(A..B)_i$ among $m_A$ nodes $A$ and $m_B$ nodes $B$ can be done in
$f_1(m_A,m_B,i)={m_A\choose {i}}{m_B\choose {i}}({i}!)^2$ ways; there are $(m_A-i)$ nodes $A$, and $(m_B-i)$ nodes $B$ left.
The same is used for choosing a chain of type $(B..A)_i$;
\item 
choosing a chain of type $(A..A)_i$ among $m_A$ nodes $A$ and $m_B$ nodes $B$ can be done in
$f_2(m_A,m_B,i)={m_A\choose {i}}{m_B\choose {i-1}}i!(i-1)!$ ways; there are $(m_A-i)$ nodes $A$, and $(m_B-(i-1))$ nodes $B$ left.
The same is used for choosing a chain of type $(B..B)_i$;
\item 
choosing a chain of type $(.A..B.)_i$ among $m_A$ nodes $A$ and $m_B$ nodes $B$ can be done in
$f_3(m_A,m_B,i)={m_A\choose {i}}{m_B\choose {i}}(i!)^2/i$ ways; there are  $(m_A-i)$ nodes $A$, and $(m_B-i)$ nodes $B$ left.
Division by $i$ is done because of symmetries - every ring of type $(.A..B.)_i$ is determined by choosing $i$ nodes of type $A$, $i$ nodes of type $B$, ordering nodes $A$ in one of $i!$ ways, ordering nodes $B$ in one of $i!$ ways, but every ordering $(A_{j1}-B_{k1}-A_{j2}-B_{k2}-\ldots A_{ji}-B_{ki})$ defines the same ring as $(A_{j2}-B_{k2}-A_{j3}-B_{k3}-\ldots A_{j1}-B_{k1})$ etc. ($i$ of them in total).
 \end{itemize}

Moreover, if $s$ is such that 
$\funt(s)=(\nbonds_{rl},\nbonds_{ba})$, then
\[
\al_{2i}(s) = {n\choose {\nbonds_{rl}}}^2\nbonds_{rl}! {n\choose {\nbonds_{ba}}}^2 \nbonds_{ba}!.
\]
If $s$ is such that
$\funt(s)=\nbonds$, then
\[
\al_{3i}(s) = \sum_{i=0}^{\nbonds}  {n\choose {i}}^2 i! {n\choose {\nbonds-i}}^2 (\nbonds-i)!.
\]

We choose $m_{rl}$ nodes of type $A$ among $n$ of them, and the same number of nodes of type $B$. 
There is $m_{rl}!$ different ways to connect them. 
We independently choose the $m_{ba}$ bonds in the same way. 

To compute $\al_{3i}(s)$, since all of the $m$ bonds can be either of type $m_{rl}$ or $m_{ba}$, we choose $i$ bonds of type $m_{rl}$ and $(m-i)$ bonds of type $m_{ba}$, for $i=0,\ldots,m$.


\subsection{Figure~\ref{fig:cs1}}
The CTMC $\{\proc_t\}$, for given one node $A$, three nodes $B$ and one node $C$ contains 
different reaction mixtures over the set of nodes 
$\{A^{1}, B^{1}, B^{2}, B^{3}, C^{1}\}$. 
For example, let $\usg$ be the reaction mixture with the set of edges $\{\{(A^1,b), (B^3,a)\}\}$.
There are three ways to apply the rule $R_2$ on $\usg$:
by embedding via function $\inst_1= 
\begin{pmatrix}
B & C\\
B^1 & C^1
\end{pmatrix}$, 
 $\inst_2= 
\begin{pmatrix}
B & C\\
B^2 & C^1
\end{pmatrix}$, or 
 $\inst_3= 
\begin{pmatrix}
B & C\\
B^3 & C^1
\end{pmatrix}$. 
If $\usg_1$ is a mixture with a set of edges $\{\{(B^{3},a),(A^{1},b)\},\{(B^2,c),(C^1,b)\}\}$ and 
$\usg_2$ is a mixture with a set of edges $\{\{(B^{3},a),(A^{1},b)\},\{(B^3,c),(C^1,b)\}\}$,
then $Q(\usg,\usg_1)=Q(\usg,\usg_2)=\rate_2$.

Note that $\funo(\usg)=(1,0,0)$, $\funo(\usg_1)=(1,1,0)$, $\funo(\usg_2)=(0,0,1)$.
Let $\usg\in \eqclass_1$, $\usg_1\in \eqclass_2$, $\usg_2\in \eqclass_3$.
By applying the Equation (\ref{eq:alphao}), 
we have
$\al_{11}(\usg)=(\frac{1!3!1!}{1!0!0!0!2!1!})^{-1}$=1/3,
$\al_{12}(\usg_1)=(\frac{1!3!1!}{1!0!1!0!2!0!})^{-1}=1/3$, and
$\al_{13}(\usg_2)=(\frac{1!3!1!}{1!1!0!0!1!0!})^{-1}=1/6$.

Moreover, since $\funt(\usg)=(1,0)$, and $\funt(\usg_1)=\funt(\usg_2)=(1,1)$, let 
$\eqcb_1, \eqcb_2\in {\allsg_2}$ be such that $\usg\in \eqcb_1$ and $\usg_1,\usg_2\in \eqcb_2$.
Then, 
$\al_{21}(\usg)=({1\choose 1} {3\choose 1} 1! {1\choose 0} {3\choose 0} 0!)^{-1}=1/3$ and
$\al_{22}(\usg_1)=\al_{22}(s_2)=({1\choose 1} {3\choose 1} 1! {1\choose 1} {3\choose 1} 1!)^{-1}=1/9$.

Finally, observing the aggregation from $\allsg_1$ to $\allsg_2$, we have that
$\al_1(\eqclass_1)=\frac{\al_{21}(\usg)}{\al_{11}(\usg)}=1$,
$\al_2(\eqclass_2)=\frac{\al_{22}(\usg_1)}{\al_{12}(\usg_1)}=1/3$, and 
$\al_2(\eqclass_3) = \frac{\al_{22}(\usg_2)}{\al_{13}(\usg_2)}=2/3$.

\subsection{Table~\ref{tab:cs3}}
\label{app:red}
In order to illustrate how powerful the presented reduction method is in comparison to the standard, species-based models, we compare the size of the state space in the species-based model, $\allsg_1$, and in the fragment-based model, $\allsg_2$.

\begin{description}
\item[Simple scaffold.]
The size of $\allsg_2$ is $(n+1)^2$: there are $n+1$ possible situations between $A$ and $B$ nodes with $0$,$1$,$\ldots$,$n$ bonds between them. 
The same holds for possible configurations between nodes of type $B$ and $C$.
Let $f(k)$ denote the number of states with $k$ copies of each of the nodes $A$, $B$ and $C$, and with no complexes of type $(ABC)$. 
If there is $0\leq i\leq k$ complexes of type $(AB)$, there can be $0\leq j\leq (k-i)$ complexes of type $(BC)$, and we thus have $f(k)=\sum_{i=0}^k (k-i+1)=\frac{(k+1)(k+2)}{2}$.
The number of complexes of type $(ABC)$ can vary from $0$ to $n$, and thus we have the total number of states in $\allsg_1$ to be
$\sum_{k=0}^n f(k)=\frac{1}{2} \sum_{k=0}^n (k^2+3k+2) = \frac{1}{2} (\sum_{k=0}^n k^2+3\sum_{k=0}^n k + 2\sum_{k=0}^n 1) =
\frac{1}{2} (n(n+1)(2n+1)/6+3n(n+1)/2 + 2(n+1))=(n+1)(n+2)(n+3)/6$.
\item[Two-sided polymerization.]
We first estimate the size of $\allsg_2$.
The value of $m_{rl}$ varies between $0$ and $n$, and the same holds for the value of $m_{ba}$.
Each state $(i,j)\in\{0,\ldots,n\}\times\{0,\ldots,n\}$ is reachable, since the bonds are created independently of each-other.
The size of the state space $\allsg_2$ is thus $(n+1)^2$.
The size of $\allsg_2$ is $2n+1$, because the value of $m$ varies between $0$ and $2n$.
Let $P(n)$ denote the number of partitions of number $n$ - number of ways of writing $n$ as a sum of positive integers.
One of the well-known asymptotics is $P(n)\approx \frac{1}{4n\sqrt{3}} e^{\pi\sqrt{\frac{2n}{3}}}$ \cite{ramanujan18}.
Consider one partition $n=n_1+\ldots+n_k$, $n_1\leq\ldots\leq n_k$,
and a state $s_1\in \allsg_1$ that counts one chain of type $(A..B)_{n_1}$, one chain of type $(A..B)_{n_2}$ etc. 
It is in $\allsg_1$, because it has exactly $n$ nodes $A$ and $n$ nodes $B$.
Therefore, the set $\allsg_1$ counts at least $P(n)$ states.
This approximation can be improved by factor three: think of the states $\usg_2$ and $\usg_3$, which are constructed of chains of type $(B..A)_i$, or $(.A..B.)_i$ instead of $(A..B)_i$.
 \end{description}

\subsection{Relation between site-graph-rewrite rules and Kappa}
\label{Kappa}

Since the main purpose of this paper is not to formally present the reduction procedure for a general rule-set,
we described the rule-based model directly as a collection of site-graph-rewrite rules, which is
a simplification with respect to standard site-graph framework of Kappa (\cite{lics2010}).
The simplification arises in three aspects. 

First, the site (protein domain) in Kappa may be \emph{internal}, in the sense that they bear an internal state encoding, for instance, post-translational modification of protein-residues such as phosphorylation, methylation, ubiquitylation - to name a few. 
Moreover, one site can simultaneously serve as a binding site, and as an internal site. 
We omit the possibility of having internal sites, but, it can be overcome: 
for example, the phosphorylation of a site can be encoded by a binding reaction to a node with a new name, for example, $Ph$. 
In order to mimic the standard unimolecular modification process by this bimolecular one, we need to ensure that the nodes of type $Ph$ are always highly abundant, that is, are not rate limiting at any time. As a side remark, we point out that in reality it takes a binding event (e.g. binding of ATP) for a modification to happen.  
If a site is both internal and binding site, another copy of the site is created, so that one site bears an internal state, and another one is a binding state.
A Kappa rule and an example of the corresponding site-graph-rewrite rule are shown in 
Figure \ref{fig:cs3}b. 

Second, each Kappa program has a predefined signature of site types and agent types, where the agent type consists of a name, and a predetermined interface (set of sites).
Each node of a `Kappa' site-graph is assigned a unique name. 
On top of that, a type function partitions all the nodes according to their agent type. 
We instead embed the information about the node type (and we also abandon the use of term `agent' in favor of `node') directly in the name of the node: a node $v^i$, $i\in\N$ is of type $v$; 
The rules are accordingly written with these generative node names.
The interface of a node type $v$ is read from the collection of site-graph-rewrite rules, as a union of all the sites which are assigned to $v$ along the rules.
Our formalism cannot specify a rule which operates over a connected site-graph with more than one node of a certain type,
but the examples which we present here do not contain such rules. 

Third, we restrict to the conserved systems -- only edges can be modified by the rules, while Kappa can specify agent birth or deletion.

Finally, it is worth noting that we define the notion of embedding in a non-standard way, through a combination of node-renaming function and sub-site-graph property.

\end{document}